\pdfoutput=1
\documentclass[11pt]{article}

\usepackage[margin=1in]{geometry}
\usepackage[T1]{fontenc}
\usepackage[utf8]{inputenc}

\usepackage{subcaption}
\usepackage[running, mathlines]{lineno}
\usepackage{braket,amsfonts,amssymb}
\usepackage{enumitem,mathtools}
\usepackage{xcolor}
\usepackage{colortbl}
\usepackage{amsmath}
\usepackage{amsthm}
\usepackage{array}
\usepackage{wrapfig}
\usepackage{pgfplots}
\usepackage{algorithmic}
\usepackage{tikz}
\usepackage{graphicx,epstopdf}
\usepackage{mathrsfs}
\usepackage{amsopn}
\usepackage{xspace}
\usepackage{bold-extra}
\usepackage[most]{tcolorbox}
\usepackage{caption}
\usepackage{color}
\usepackage{etoolbox}
\usepackage{titlesec}
\usepackage[colorlinks=true,linkcolor=blue,citecolor=blue,urlcolor=blue]{hyperref}
\usepackage[nameinlink,capitalize]{cleveref}
\usepackage{listings}

\AtBeginDocument{
\addtolength{\abovedisplayskip}{-0.8ex}
\addtolength{\abovedisplayshortskip}{-0.8ex}
\addtolength{\belowdisplayskip}{-0.8ex}
\addtolength{\belowdisplayshortskip}{-0.8ex}
}

\makeatletter
\def\thm@space@setup{\thm@preskip=3pt \thm@postskip=3pt}
\makeatother

\titlespacing*{\section}{0pt}{0.6\baselineskip}{0.2\baselineskip}
\titlespacing*{\subsection}{0pt}{0.6\baselineskip}{0.2\baselineskip}

\numberwithin{equation}{section}
\numberwithin{table}{section}
\numberwithin{figure}{section}

\newcommand{\EQ}{\begin{equation}}
\newcommand{\EN}{\end{equation}}
\newcommand{\EQS}{\begin{equation*}}
\newcommand{\ENS}{\end{equation*}}
\newcommand{\EQA}{\begin{eqnarray}}
\newcommand{\ENA}{\end{eqnarray}}
\newcommand{\EQAS}{\begin{eqnarray*}}
\newcommand{\ENAS}{\end{eqnarray*}}
\newcommand{\AL}{\begin{align}}
\newcommand{\AN}{\end{align}}
\newcommand{\ALS}{\begin{align*}}
\newcommand{\ANS}{\end{align*}}

\newcommand{\myblue}{\color{black}}

\definecolor{lightblue}{rgb}{0.3, 0.0, 0.9}

\newcommand{\Ebb}{\mathbb{E}}

\newcommand{\Rbb}{\mathbb{R}}

\newtheorem{theorem}{Theorem}[section]
\newtheorem{lemma}[theorem]{Lemma}

\newtheorem{assumption}{Assumption}[section]
\newtheorem{definition}{Definition}[section]

\theoremstyle{remark}
\newtheorem{remark}[theorem]{Remark}

\crefname{hypothesis}{Hypothesis}{Hypotheses}
\Crefname{hypothesis}{Hypothesis}{Hypotheses}
\Crefname{ALC@unique}{Line}{Lines}

\colorlet{texcscolor}{blue!50!black}
\colorlet{texemcolor}{red!70!black}
\colorlet{texpreamble}{red!70!black}
\colorlet{codebackground}{black!25!white!25}

\lstdefinestyle{siamlatex}{%
  style=tcblatex,
  texcsstyle=*\color{texcscolor},
  texcsstyle=[2]\color{texemcolor},
  keywordstyle=[2]\color{texemcolor},
  moretexcs={cref,Cref,maketitle,mathcal,text,email,url},
}

\tcbset{%
  colframe=black!75!white!75,
  coltitle=white,
  colback=codebackground,
  colbacklower=white,
  fonttitle=\bfseries,
  arc=0pt,outer arc=0pt,
  top=1pt,bottom=1pt,left=1mm,right=1mm,middle=1mm,boxsep=1mm,
  leftrule=0.3mm,rightrule=0.3mm,toprule=0.3mm,bottomrule=0.3mm,
  listing options={style=siamlatex}
}

\newtcblisting[use counter=example]{example}[2][]{%
  title={Example~\thetcbcounter: #2},#1}

\newtcbinputlisting[use counter=example]{\examplefile}[3][]{%
  title={Example~\thetcbcounter: #2},listing file={#3},#1}

\DeclareTotalTCBox{\code}{ v O{} }
{
  fontupper=\ttfamily\color{black},
  nobeforeafter,
  tcbox raise base,
  colback=codebackground,colframe=white,
  top=0pt,bottom=0pt,left=0mm,right=0mm,
  leftrule=0pt,rightrule=0pt,toprule=0mm,bottomrule=0mm,
  boxsep=0.5mm,
  #2}{#1}

\patchcmd\newpage{\vfil}{}{}{}
\title{Convergence of Neural Network Policies for Risk--Reward Optimization}

\author{
Chang Chen\thanks{School of Mathematics and Physics, The University of Queensland, St Lucia, Brisbane 4072, Australia (\href{mailto:chang.chen1@student.uq.edu.au}{chang.chen1@student.uq.edu.au}).}
\and
Duy-Minh Dang\thanks{School of Mathematics and Physics, The University of Queensland, St Lucia, Brisbane 4072, Australia (\href{mailto:duyminh.dang@uq.edu.au}{duyminh.dang@uq.edu.au}).}
}

\date{}

\begin{document}
\maketitle

\begin{abstract}
We develop a neural-network framework for multi-period risk--reward
stochastic control problems with constrained two-step feedback policies that
may be discontinuous in the state.
We allow a broad class of objectives built on a
finite-dimensional performance vector, including terminal and path-dependent
statistics, with risk functionals admitting auxiliary-variable optimization
representations (e.g.\ Conditional Value-at-Risk and buffered probability of
exceedance) and optional moment dependence.
Our approach parametrizes the two-step policy using two coupled feedforward
networks with constraint-enforcing output layers, reducing the constrained control
problem to unconstrained training over network parameters.
Under mild regularity conditions, we prove that the
empirical optimum of the NN-parametrized objective
converges in probability to the true optimal value as network capacity and
training sample size increase. The proof is modular, separating
policy approximation, propagation through the controlled recursion, and
preservation under the scalarized risk--reward objective.
Numerical experiments confirm the predicted convergence-in-probability behavior, show close agreement between learned and reference control heat maps, and demonstrate out-of-sample robustness on a large \mbox{independent scenario set.}
\end{abstract}

\vspace{0.5em}
\noindent\textbf{Keywords.}
neural networks, policy approximation, risk--reward optimization, convergence analysis

\medskip
\noindent\textbf{MSC2020.}
93E20, 68T07, 90C15, 62G05

\section{Introduction}
\label{sec:intro}
Discrete-intervention stochastic control problems arise whenever decisions are made at a finite set of intervention times and the system evolves stochastically between interventions. In many applications, decisions are subject to pointwise constraints (e.g.\ feasibility, budgets, or operational limits),
and the objective function encodes a trade-off between a notion of reward and a notion of risk. This includes settings with two-step constrained interventions at each time (e.g.\ a pre-decision adjustment followed by a post-decision allocation), as well as terminal and discrete-time path-dependent performance criteria. Such problems arise broadly in finance and insurance, economics, and engineering.

Neural networks (NNs) provide a flexible framework for approximating feedback policies in settings where dynamic-programming discretizations become expensive or infeasible. An increasingly studied approach is to parameterize the policy directly and optimize an empirical objective via stochastic gradient methods. This idea underlies a rapidly developing literature on NN-based policy approximation for stochastic control and related problems \cite{buehler2019deep, li2019data, reppen2023deepstochastic, reppen2023deep, vanstaden2024global, hure2021deep, bachouch2022deep, han2016deep}.
Following the terminology in \cite{hu2024recent}, these approaches can be divided into two categories: global-in-time and local-in-time. In a global-in-time approach,
the control over all decision times is determined by solving a single training problem at the initial time (even if the parametrization uses multiple time-indexed subnetworks), as in, e.g.\ \cite{buehler2019deep, li2019data, reppen2023deepstochastic, reppen2023deep, vanstaden2024global}. In contrast, local-in-time methods follow a backward-induction structure and train time-indexed approximations sequentially, typically using a separate network (or parameter set) at each decision time \cite{hure2021deep, bachouch2022deep}. In this sense, local-in-time implementations are often described as stacked NN schemes, since they stack one approximator per decision time and fit them step-by-step (see also \cite{tsang2020deep}). From a broader viewpoint, both paradigms are instances of policy function approximation for stochastic control \cite{powell2023universal}.

%A central question for these approaches concerns the convergence pipeline: how approximation of feedback maps by NNs propagates through the controlled state recursion and, in turn, through a risk--reward objective, and whether the resulting empirical training problem is consistent, i.e.\ converges to the true optimum as the training sample size and NN capacity increase.
%A key technical difficulty is that learned policies are evaluated at moving inputs, since the controlled state is itself generated by the learned policy.
%A standard strategy in existing convergence pipelines is to assume global continuity of the optimal feedback rule in the state variable and invoke uniform approximation over the feature domain. A representative example is the discrete-time global-in-time convergence analysis in \cite{vanstaden2024global}, which uses these assumptions to control evaluation at network-dependent state inputs.

A central question for these approaches concerns the convergence pipeline: how approximation of feedback maps by NNs propagates through the controlled state recursion and, in turn, through a risk--reward objective, and whether the resulting empirical training problem is consistent, i.e.\ converges to the true optimum as the training sample size and NN capacity increase.
A key technical difficulty is that learned policies are evaluated at moving inputs, since the controlled state is itself generated by the learned policy.
{\myblue{A convenient sufficient route in discrete-time convergence analyses is to assume global continuity of the optimal feedback rule in the state variable and to invoke uniform approximation over the relevant feature domain; see, for example, \cite{vanstaden2024global}.}}

While appropriate when continuity is structurally justified, this strategy does not directly accommodate constrained controls that naturally exhibit discontinuities, a common feature of practical intervention problems. Even in low-dimensional models, binding constraints can induce threshold/bang--bang feedback rules, so sup-norm uniform approximation is not the right notion for discontinuous targets
and continuity at moving inputs can fail on the discontinuity set.

This paper develops a convergence framework for discrete-intervention risk--reward control that remains valid in the presence of such discontinuities, and separates the approximation, propagation, and objective layers of the argument. Our setting accommodates a broad class of scalarized risk--reward objectives, including path-dependent criteria, auxiliary-variable (multi-level) risk representations,
as well as moment dependence.  As is standard in global-in-time policy approximation,
our analysis is formulated at the initial time. For objectives that are separable in the dynamic-programming sense, this yields the time-consistent optimal strategy; for nonseparable risk--reward criteria, it yields the corresponding optimal \mbox{pre-commitment strategy.}

Our main contributions are as follows.
\begin{itemize}[leftmargin=*, noitemsep, topsep=2pt]
\item
We formulate a discrete-intervention control problem with a two-step feedback policy evaluated at pre- and post-decision times. This captures settings where a pre-decision adjustment (e.g.\ withdrawal/consumption/liquidation) is coupled with a post-decision allocation step, with both actions subject to pointwise constraints.

\item
We represent a broad class of risk-reward objectives via a finite-dimensional performance vector extracted from the controlled recursion, allowing terminal and discrete-time path-dependent statistics. The risk functional can be specified through auxiliary-variable optimization representations (e.g.\ Conditional Value-at-Risk and buffered probability of exceedance), as well as path-based measures such as realized quadratic variation, with optional higher-moment dependence. This yields a modular objective class within our convergence framework.

\item
We parameterize the two policy components by feedforward NNs equipped with constraint-enforcing output layers (e.g.\ interval and simplex maps), so feasibility holds by construction while the optimization is unconstrained in the NN weights.

\item
We replace global continuity requirements with a weaker ``null discontinuity'' condition: the optimal feedback maps may be discontinuous provided their discontinuity sets are hit with probability zero under the optimal controlled state at intervention times. Using a probabilistic moving-input stability argument (based on the Portmanteau theorem \cite{Billingsley1999}), we propagate NN approximation through the controlled recursion without requiring global continuity of the optimizer.

\item
Under compact-domain regularity conditions, we prove that the empirical optimum of the NN-parametrized objective converges in probability to the true optimal value as NN capacity and training sample size increase. The argument is modular, separating (i) approximation within the admissible policy class, (ii) propagation through the controlled recursion under moving-input stability, (iii) preservation under a general scalarized risk--reward functional, and (iv) a uniform law of large numbers for the empirical objective.
\end{itemize}
Numerical experiments on a representative Defined Contribution decumulation problem, in which withdrawals and allocations are constrained and the computed optimal withdrawal policy exhibits bang--bang structure, are benchmarked against a provably convergent low-dimensional grid-based reference value.
In addition to the predicted convergence-in-probability behavior as NN capacity and training sample size increase, the learned withdrawal and allocation heat maps show excellent agreement with the reference policies, and these conclusions remain stable under evaluation on a large independent out-of-sample scenario set.

\section{Problem formulation}
\label{sec:formulation}

\subsection{Probability space and intervention times}
We work on a complete filtered probability space
$\big(\Omega,\mathcal{F},\{\mathcal{F}_t\}_{0\le t\le T},\mathbb{P}\big)$,
satisfying the usual conditions on a finite horizon $T>0$.
Within this probabilistic framework, decisions are made at a predetermined, equally spaced set of intervention times
\begin{equation}
\label{eq: T_M}
\mathcal{T}=\{t_m:\ t_m=m\Delta t,\ m=0,\ldots,M\},\qquad \Delta t=T/M,
\end{equation}
where $t_0=0$ is the initial time.\footnote{Equal spacing is assumed for notational simplicity; the analysis extends directly to any fixed finite set of intervention times.}
At each $t_m\in\mathcal{T}$, we distinguish a pre-decision time $t_m^-$ and a post-decision time $t_m^+$.
With the shorthand $t^-:=t-\epsilon$ and $t^+:=t+\epsilon$ (as $\epsilon\to 0^+$), we write
$f_{m^-}:=\lim_{\epsilon\to 0^+} f(t_m-\epsilon)$ and
$f_{m^+}:=\lim_{\epsilon\to 0^+} f(t_m+\epsilon)$.

\subsection{Exogenous input process}
Fix an integer $d_a\ge 1$.
Let $Y_m:=(Y_m^{(i)})_{i=1}^{d_a}\in\Rbb^{d_a}$ denote an exogenous input (or shock) vector at each $t_m\in\mathcal{T}$.
Then $\{Y_m\}_{m=1}^{M}$ defines an $\Rbb^{d_a}$-valued discrete-time stochastic process adapted to
$\{\mathcal{F}_{t_m^-}\}_{m=1}^{M}$, where $Y_m$ is realized over $[t_{m-1}^+,t_m^-]$ and observed at time $t_m^-$.
We write this process as $Y$, where
\begin{equation}
\label{eq: Y_R}
%\mathbf{Y}:=\{Y_m\}_{m=1}^{M},\qquad Y_m:=(Y_m^i)_{i=1}^{d_a}.
Y:=\{Y_m\}_{m=1}^{M},\qquad Y_m:=(Y_m^{(i)})_{i=1}^{d_a}.
\end{equation}
Assume $\Ebb[\|Y_m\|]<\infty$ for each $m\in \{1,\ldots,M\}$,
where $\|\cdot\|$ denotes a \mbox{fixed norm on $\Rbb^{d_a}$.}

\subsection{State, admissible controls, and controlled dynamics}
We consider a scalar controlled state process $\{W(t)\}_{0\le t\le T}$.
As feature (state) variables, we use time and the current state value,
\[
\phi(t):=(t,W(t)),\qquad \mathcal{D}_\phi:=[0,T]\times\mathbb{R}.
\]
A two-step feedback control is a pair $\mathcal{P}=(q,p)$ of Borel measurable maps
\[
q:\mathcal{D}_\phi\to\mathbb{R},
\qquad
p:\mathcal{D}_\phi\to\mathbb{R}^{d_a},
\]
where $q$ is applied at the pre-decision times $t_m^-$ for $m=0,\ldots,M$, and
$p$ is applied at the post-decision times $t_m^+$ for $m=0,\ldots,M-1$
(i.e.\ no post-decision action is taken at $t_M^+$). In many applications, $q$ represents a pre-decision state adjustment (injection or extraction) and $p$ an allocation of a scalar resource among $d_a$ components, but our analysis only uses measurability and the pointwise constraints below.

\medskip
\noindent\textit{Pointwise constraints.}
We impose (i) an interval-type constraint on the pre-decision action $q$ and (ii) a simplex-type constraint on the post-decision action $p$.
For given constants $q_{\min}\le q_{\max}$, define the admissible pre-decision set map
\begin{equation}
\label{eq: Z_q}
\mathcal{Z}_q(w)
=
\begin{cases}
\left[q_{\min},\,q_{\max}\right], & \text{if } w\ge q_{\max},\\[2pt]
\left[q_{\min},\,w\right], & \text{if } q_{\min}< w< q_{\max},\\[2pt]
\left\{q_{\min}\right\}, & \text{if } w\le q_{\min}.
\end{cases}
\end{equation}
Admissible post-decision actions lie in the simplex
\begin{equation}
\label{eq: Z_p}
\mathcal{Z}_p=\big\{z\in\mathbb{R}^{d_a}:\ z_i\ge 0,\ \sum_{i=1}^{d_a}z_i=1\big\}.
\end{equation}
Accordingly, define the admissible policy classes
\begin{align}
\mathcal{G}_{q}
&=
\left\{
q:\mathcal{D}_{\phi}\to \mathbb{R}
\ \middle|\
q \text{ is Borel measurable and } q(t,w)\in \mathcal{Z}_q(w)\ \forall (t,w)\in\mathcal{D}_\phi
\right\},
\label{eq: G_q}
\\
\mathcal{G}_{p}
&=
\left\{
p:\mathcal{D}_{\phi}\to \mathbb{R}^{d_a}
\ \middle|\
p \text{ is Borel measurable and } p(t,w)\in \mathcal{Z}_{p}\ \forall (t,w)\in\mathcal{D}_\phi
\right\},
\label{eq: G_p}
\end{align}
and the admissible control pairs
\begin{equation}
\label{eq: G}
\mathcal{G}
=
\left\{
\mathcal{P}=(q,p):\ q\in\mathcal{G}_{q},\ p\in\mathcal{G}_{p}
\right\}.
\end{equation}

\noindent\textit{Controlled recursion at intervention times.}
Given $\mathcal{P}\in\mathcal{G}$ and the exogenous input process $Y$, let
$\{W(t;\mathcal{P},Y)\}_{0\le t\le T}$ denote the induced controlled state process.
When the dependence on $(\mathcal{P},Y)$ is understood, we write $W(t) \equiv W(t;\mathcal{P},Y)$. We also write $W(T^+;\mathcal{P},Y)$ for the terminal post-decision state immediately after the final
pre-decision action at time $T$ (if any).

Fix measurable update maps
\[
U_q:[0,T]\times\mathbb{R}\times\mathbb{R}\to\mathbb{R},
\qquad
U_p:[0,T]\times\mathbb{R}\times\mathbb{R}^{d_a}\times\mathbb{R}^{d_a}\to\mathbb{R}.
\]
Given $\mathcal{P}\in\mathcal{G}$ and $Y$, we define the induced controlled state sequence recursively~by
\begin{equation}
\label{eq:rr_skeleton_updates}
W(t_m^+;\mathcal{P},Y)
=
U_q\left(t_m,\;W(t_m^-;\mathcal{P},Y),\;q\left(t_m^-,W(t_m^-;\mathcal{P},Y)\right)\right),
\end{equation}
for $m=0,\ldots,M$, and
\begin{equation}
\label{eq:rr_skeleton_evolution}
W(t_{m+1}^-;\mathcal{P},Y)
=
U_p\left(t_m,\;W(t_m^+;\mathcal{P},Y),\;p\left(t_m^+,W(t_m^+;\mathcal{P},Y)\right),\;Y_{m+1}\right),
\end{equation}
for $m=0,\ldots,M-1$, with initial condition $W(t_0^-)=w_0$.
We interpret $W(T^+;\mathcal{P},Y)$ as the terminal post-decision state after the final pre-decision action at $T=t_M$.
We write $W(t_m^\pm):=W(t_m^\pm;\mathcal{P},Y)$ whenever no confusion can arise.
\begin{remark}[Well-posedness of the controlled recursion]
\label{rem:well_posed_recursion}
Fix $\mathcal{P}\in\mathcal{G}$ and $Y$.
Starting from $W(t_0^-)=w_0$, the update rules \eqref{eq:rr_skeleton_updates}--\eqref{eq:rr_skeleton_evolution}
define recursively a unique finite sequence $\{W(t_m^-),W(t_m^+)\}_{m=0}^{M}$ whenever the maps $U_q$ and $U_p$ are measurable.
The convergence analysis later only uses measurability and continuity of these updates on bounded sets, and it allows piecewise-defined dynamics (e.g.\ regime switching) provided the realized recursion is continuous in the current state and action variables on the relevant bounded domain.
\end{remark}
\subsection{Risk--reward objective}
Throughout, $\mathbb{E}^{w_0,t_0^-}_{\mathcal{P}}[\cdot]$ denotes expectation under $\mathbb{P}$ with respect to the probability law induced by applying the control pair $\mathcal{P}$, conditional on the initial condition $W(t_0^-)=w_0$. The only source of randomness is the exogenous process $Y$.

Because the horizon is finite and actions are applied only at the finite intervention set $\mathcal{T}$, we consider a broad class of risk--reward objectives that can be expressed in terms of a finite-dimensional random performance vector constructed from the controlled state sequence and the realized actions at the intervention times.
Specifically, for any $\mathcal{P}$ and $Y$, we define a finite-dimensional performance vector
$S(\mathcal{P},Y)\in\mathbb{R}^d$ by collecting a chosen subset of the controlled state and action variables at the intervention times, for example
\begin{equation}
\label{eq:rr_perf_vector}
\begin{aligned}
S(\mathcal{P},Y)
:=
\Big(&
\big(W(t_m^-;\mathcal{P},Y)\big)_{m=0}^M,\;
\big(W(t_m^+;\mathcal{P},Y)\big)_{m=0}^M,\;
\\
&
\big(q(t_m^-,W(t_m^-;\mathcal{P},Y))\big)_{m=0}^M,\;
\big(p(t_m^+,W(t_m^+;\mathcal{P},Y))\big)_{m=0}^{M-1}
\Big),
\end{aligned}
\end{equation}
where $d<\infty$ depends on the selected components (in particular, $d\le 3(M+1)+M d_a$ for the collection displayed in \eqref{eq:rr_perf_vector}).\footnote{The specific contents of $S$ can be adapted to the application; the only requirement for the convergence pipeline is that $S$ is a finite-dimensional measurable functional of the controlled recursion.}

%\medskip
%\noindent\textit{Reward.}
\subsubsection{Reward}
Given the performance vector $S(\mathcal{P},Y)$, the reward component of the objective is defined as the expected value of a measurable function of $S$.
Let $\mathcal{R}:\mathbb{R}^d\to\mathbb{R}$ be measurable and define the reward functional
\begin{equation}
\label{eq:rr_reward}
\mathcal{J}_{\mathcal{R}}(w_0,t_0^-;\mathcal{P})
:=
\mathbb{E}^{w_0,t_0^-}_{\mathcal{P}}\!\big[\,\mathcal{R}(S(\mathcal{P},Y))\,\big].
\end{equation}
Here $S(\mathcal{P},Y)\in\mathbb{R}^d$ is a user-chosen finite-dimensional vector of terminal and/or path-dependent statistics of the controlled state/action sequence (e.g.\ terminal state and selected intermediate values).
This template covers, for example:
\begin{itemize}[noitemsep, topsep=2pt, leftmargin=*]
\item \emph{Terminal rewards: }
Choose $S(\mathcal{P},Y):=W(T^+)\in\mathbb{R}$.
Let $\mathcal{R}:\mathbb{R}\to\mathbb{R}$ be defined by $\mathcal{R}(s)=U(s)$,
where $U:\mathbb{R}\to\mathbb{R}$ is a given measurable payoff/utility function.
Then $\mathcal{R}\big(S(\mathcal{P},Y)\big) =
U\!\left(W(T^+)\right)$.
As a special case, taking $U(s)=s$ yields the linear terminal reward
$\mathcal{R}\big(S(\mathcal{P},Y)\big)=W(T^+)$.

\item \emph{Cumulative adjustment rewards:}
Choose $S(\mathcal{P},Y):=\left(q\!\left(t_m^-,W(t_m^-)\right)\right)_{m = 0}^M\in\mathbb{R}^{M+1}$.
Let $\mathcal{R}:\mathbb{R}^{M+1}\to\mathbb{R}$ be defined by
$\mathcal{R}(s)=\sum_{m=0}^M s_m$.
Then
$\mathcal{R}\big(S(\mathcal{P},Y)\big)=\sum_{m=0}^M q_m(\cdot)$.
In a pension setting, $q_m(\cdot)$ may represent the withdrawal at time $t_m$, so $\mathbb{E}\!\big[\sum_{m=0}^M q_m(\cdot)\big]$ is the cumulative expected withdrawals over the horizon.

\item \emph{Quadratic / one-sided quadratic criteria \cite{zhou2000, DM2016semi}:}
Choose $S(\mathcal{P},Y):=W(T^+)\in\mathbb{R}$.
Let $\mathcal{R}:\mathbb{R}\to\mathbb{R}$ be defined by
$\mathcal{R}(s)=-(s-\kappa)^2$ or
$\mathcal{R}(s)=-(\min\{s-\kappa,0\})^2+\lambda s$,
with target level $\kappa\in\mathbb{R}$ and weight $\lambda\ge 0$.
Then $\mathcal{R}\big(S(\mathcal{P},Y)\big)
=-(W(T^+)-\kappa)^2$
or
$\mathcal{R}\big(S(\mathcal{P},Y)\big)
=-(\min\{W(T^+)-\kappa,0\})^2+\lambda\,W(T^+)$.
\end{itemize}

%\medskip
%\noindent\textit{Risk via an auxiliary variable (optional moment dependence).}
\subsubsection{Risk}
We represent a broad class of risk measures using an auxiliary variable $\xi\in\Xi$ (e.g.\ a threshold in tail-risk criteria), where $\Xi\subseteq\mathbb{R}$ is an auxiliary-variable domain. Some risk measures also depend on moments of the performance vector; to accommodate this, we allow an optional moment mapping. Specifically, let $\Psi:\mathbb{R}^d\to\mathbb{R}^{d_\Psi}$ be measurable and define the (finite-dimensional) moment vector
\begin{equation}
\label{eq:rr_moment_vector}
\bar{S}(\mathcal{P})
:=
\mathbb{E}^{w_0,t_0^-}_{\mathcal{P}}\!\big[\,\Psi(S(\mathcal{P},Y))\,\big]\in\mathbb{R}^{d_\Psi}.
\end{equation}
(If moment-dependence is not needed, take $d_\Psi=0$ and suppress $\bar{S}$ below.)
Let $\mathcal{L}:\Xi\times\mathbb{R}^d\times\mathbb{R}^{d_\Psi}\to\mathbb{R}$ be measurable. Define the risk functional by
\begin{equation}
\label{eq:rr_risk}
\mathcal{J}_{\mathcal{L}}(w_0,t_0^-;\mathcal{P})
:=
\sup_{\xi\in\Xi}\,
\mathbb{E}^{w_0,t_0^-}_{\mathcal{P}}\!\big[\,\mathcal{L}(\xi,S(\mathcal{P},Y),\bar{S}(\mathcal{P}))\,\big].
\end{equation}
This template includes, for example:
\begin{itemize}[noitemsep, topsep=2pt, leftmargin=*]
\item \emph{Rockafellar--Uryasev CVaR of the terminal state \cite{RT2000}:}
\mbox{Choose $S(\mathcal{P},Y):=W(T^+)\in\mathbb{R}$} and take $\Xi=\mathbb{R}$.
No moment dependence is needed (take $d_\Psi=0$ and suppress $\bar S$).
Let $\mathcal{L}:\Xi\times\mathbb{R}\to\mathbb{R}$ be defined by
$\mathcal{L}(\xi,s)=\xi+\frac{1}{\alpha}\min\{s-\xi,0\}$ for $\alpha\in(0,1)$.
Then
$\mathcal{J}_{\mathcal{L}}(w_0,t_0^-;\mathcal{P})
=
\sup_{\xi\in\mathbb{R}}\,
\Ebb^{w_0,t_0^-}_{\mathcal{P}}\!\left[
\xi+\frac{1}{\alpha}\min\!\left\{W(T^+)-\xi,\,0\right\}
\right]$, corresponding to $\mathrm{CVaR}_{\alpha}\!\big(W(T^+)\big)$ (under the gain-based convention).
%Combining this choice with the linear terminal reward $\mathcal{R}(s)=s$ yields a mean--CVaR objective.

\item \emph{Buffered probability of exceedance (bPoE) \cite{dang2026multi, bpoe2018}:}
Choose $S(\mathcal{P},Y):=W(T^+)\in\mathbb{R}$ and fix a threshold level $D\in\mathbb{R}$.
Take $\Xi:=(D,\infty)$ and no moment dependence (set $d_\Psi=0$).
Define $\mathcal{L}:\Xi\times\mathbb{R}\to\mathbb{R}$ by
$\mathcal{L}(\xi,s)
:=
-\max\!\left(
1-\frac{s-D}{\xi-D},\,0
\right)$ for $\xi\in\Xi$.
Then
$\mathcal{J}_{\mathcal{L}}(w_0,t_0^-;\mathcal{P})
=
\sup_{\xi\in(D,\infty)}
\Ebb^{w_0,t_0^-}_{\mathcal{P}}\!\big[
-\max\!\big(
1-\tfrac{W(T^+)-D}{\xi-D},\,0
\big)
\big]$, which corresponds to $-\mathrm{bPoE}_{D}(W(T^+))$ via the standard $\inf$-representation in \cite{dang2026multi,bpoe2018}.

\item \emph{Quadratic variation risk \cite{PMVS2019}:}
Choose $S(\mathcal{P},Y):=\big(W(t_m^+)\big)_{m=0}^{M-1}\,\cup\,\big(W(t_{m+1}^-)\big)_{m=0}^{M-1}\in\mathbb{R}^{2M}$, i.e.,
$S(\mathcal{P},Y)=(s_0,\ldots,s_{2M-1})$ with $s_m=W(t_m^+)$ and $s_{M+m}=W(t_{m+1}^-)$ for $m=0,\ldots,M-1$.
Take $\Xi=\{0\}$ and  set $d_\Psi=0$.
Define $\mathcal{L}:\Xi\times\mathbb{R}^{2M}\to\mathbb{R}$ by
$\mathcal{L}(0,s)=-\sum_{m=0}^{M-1}\big(s_{M+m}-s_m\big)^2$.
Then $\mathcal{J}_{\mathcal{L}}(w_0,t_0^-;\mathcal{P})
=\Ebb^{w_0,t_0^-}_{\mathcal{P}}\!\big[\mathcal{L}(0,S(\mathcal{P},Y))\big]$.

\item \emph{Variance (via an auxiliary variable):}
Choose $S(\mathcal{P},Y):=W(T^+)\in\mathbb{R}$ and take $\Xi=\mathbb{R}$.
Let $\mathcal{L}:\Xi\times\mathbb{R}\to\mathbb{R}$ be defined by $\mathcal{L}(\xi,s)=-(s-\xi)^2$.
Then
\[
\mathcal{J}_{\mathcal{L}}(w_0,t_0^-;\mathcal{P})
=
\sup_{\xi\in\mathbb{R}}\Ebb^{w_0,t_0^-}_{\mathcal{P}}\!\big[-(W(T^+)-\xi)^2\big]
%=
%-\Ebb^{w_0,t_0^-}_{\mathcal{P}}\!\big[\big(W(T^+)-
%\Ebb^{w_0,t_0^-}_{\mathcal{P}}[W(T^+)]\big)^2\big]
=
-\mathrm{Var}\!\big(W(T^+)\big).
\]

\item \emph{Semi-variance around the mean:}
Choose $S(\mathcal{P},Y):=W(T^+)\in\mathbb{R}$ and take $\Xi=\{0\}$.
Let $\Psi:\mathbb{R}\to\mathbb{R}$ be defined by $\Psi(s)=s$, so that
$\bar S(\mathcal{P})=\Ebb^{w_0,t_0^-}_{\mathcal{P}}\!\big[W(T^+)\big]$.
Define $\mathcal{L}:\Xi\times\mathbb{R}\times\mathbb{R}\to\mathbb{R}$ by
$\mathcal{L}(0,s,\bar s)= -\min\{s-\bar s,0\}^2$.
Then
\[
\mathcal{J}_{\mathcal{L}}(w_0,t_0^-;\mathcal{P})
=
\Ebb^{w_0,t_0^-}_{\mathcal{P}}\!\left[-\min\!\left\{W(T^+)-\bar S(\mathcal{P}),\,0\right\}^2\right],
\quad
\bar S(\mathcal{P})=\Ebb^{w_0,t_0^-}_{\mathcal{P}}\!\big[W(T^+)\big].
\]

\end{itemize}

\subsubsection{Scalarized risk--reward objective}
Fix a scalarization parameter $\gamma>0$ and define the scalarized criterion function
\[
\mathcal{H}(\xi,s,\bar s):=\mathcal{R}(s)+\gamma\,\mathcal{L}(\xi,s,\bar s),
\qquad
(\xi,s,\bar s)\in\Xi\times\mathbb{R}^d\times\mathbb{R}^{d_\Psi}.
\]
For $(\xi,\mathcal{P})\in\Xi\times\mathcal{G}$, define the corresponding scalarized objective
\begin{equation}
\label{eq:rr_objective}
V\left(w_0,t_0^-;\xi,\mathcal{P}\right)
:=
\mathbb{E}^{w_0,t_0^-}_{\mathcal{P}}\left[
\mathcal{H}\left(\xi,S(\mathcal{P},Y),\bar{S}(\mathcal{P})\right)
\right],
\end{equation}
and the value function
\begin{equation}
\label{eq:rr_value}
V(w_0,t_0^-)
:=
\sup_{\mathcal{P}\in\mathcal{G},\ \xi\in\Xi}\,
V\left(w_0,t_0^-;\xi,\mathcal{P}\right).
\end{equation}
\begin{remark}[Pre-commitment vs.\ time-consistent formulations]
The criterion \eqref{eq:rr_objective}--\eqref{eq:rr_value} is posed at the initial time $t_0^-$: the decision maker selects $(\xi,\mathcal{P})\in\Xi\times\mathcal{G}$ to maximize $V(w_0,t_0^-;\xi,\mathcal{P})$ and then applies the resulting feedback maps $\mathcal{P}=(q,p)$ at all subsequent intervention times.

If the objective is separable in the dynamic-programming sense (i.e.\ it admits a recursive Bellman representation), this formulation coincides with the usual time-consistent optimal strategy. Otherwise, the objective is time-inconsistent and the solution corresponds to the optimal pre-commitment strategy.
\end{remark}
%\begin{remark}[Pre-commitment formulation]
%The criterion \eqref{eq:rr_objective}--\eqref{eq:rr_value} is a \emph{pre-commitment} risk--reward objective: at the initial time $t_0^-$, the decision maker selects $(\xi,\mathcal{P})\in\Xi\times\mathcal{G}$ to maximize $V(w_0,t_0^-;\xi,\mathcal{P})$, and then commits to the resulting feedback policy $\mathcal{P}$ at all subsequent intervention times. In particular, the actions applied at later times are determined by the fixed feedback maps in $\mathcal{P}$, rather than by re-optimizing the objective along the realized trajectory.
%\end{remark}

\subsection{Standing assumptions for the risk--reward control problem}
\label{ssc:control_assump}
This subsection collects regularity conditions on the risk-reward control problem
\eqref{eq:rr_objective}-\eqref{eq:rr_value}
used in the convergence analysis.
\begin{assumption}[Regularity conditions]
\label{ass:rr_regularity}
The following conditions hold.
\begin{enumerate}[noitemsep, topsep=2pt, label=(R\arabic*), leftmargin=*]
\item {(Bounded state.)}
There exist finite constants $w_{\min} < w_{\max}$ such that, for every admissible $\mathcal{P}\in\mathcal{G}$,
\[
w_{\min}\le W(t_m^\pm;\mathcal{P},Y)\le w_{\max}
\quad\text{a.s.},\qquad m=0,\ldots,M.
\]

\item (Continuity of the update maps.)
The map $U_q$ is continuous on
\[
[0,T]\times[w_{\min},w_{\max}]\times[q_{\min},q_{\max}]
\]
and, for every $r>0$, the map $U_p$ is continuous on
\[
[0,T]\times[w_{\min},w_{\max}]\times \mathcal{Z}_p \times \{y\in\mathbb{R}^{d_a}:\ \|y\|\le r\}.
\]

\item (Existence of an optimizer.)
There exist $(\mathcal{P}^\ast,\xi^\ast)\in\mathcal{G}\times\Xi$ attaining the supremum in \eqref{eq:rr_value}, with $\mathcal{P}^\ast=(q^\ast,p^\ast)$.

\item (Almost-sure continuity under the optimal policy.)
For each $m$, define
\begin{align*}
D_{q,m}
&:= \big\{ w \in [w_{\min},w_{\max}] : w \mapsto q^\ast(t_m^{-},w)\ \text{is discontinuous at } w \big\},
%\label{eq:D_qm_def}
\\
& \qquad \qquad \qquad \qquad \qquad \qquad m=0,\ldots,M,
\nonumber
\\
D_{p,m}
&:= \big\{ w \in [w_{\min},w_{\max}] : w \mapsto p^\ast(t_m^{+},w)\ \text{is discontinuous at } w \big\},
%\label{eq:D_pm_def}
\\
& \qquad \qquad \qquad \qquad \qquad \qquad m=0,\ldots,M-1,
\nonumber
\end{align*}
with discontinuity for $p^\ast$ understood componentwise.
Assume that these discontinuity sets are
$\mathbb{P}$-null under the optimal state at the intervention times, i.e.\
\begin{align*}
\mathbb{P}\!\left(W(t_m^-;\mathcal{P}^\ast,Y)\in D_{q,m}\right)&=0,\quad m=0,\ldots,M, \quad \text{ and }
\\
\mathbb{P}\!\left(W(t_m^+;\mathcal{P}^\ast,Y)\in D_{p,m}\right)&=0,\quad m=0,\ldots,M-1.
\end{align*}

\item (Moment map regularity.)
Assume $\Psi$ is continuous on $\mathcal{Z}$,
where $\mathcal{Z}\subset\mathbb{R}^d$
is a compact set such that the vector $S(\mathcal{P},Y)\in\mathcal{Z}$ a.s.\ for all $\mathcal{P}\in\mathcal{G}$
(see Section~\ref{ssc:rr_regularity_discussion} for a sufficient construction of $\mathcal{Z}$).

\item (Compactness of the auxiliary-variable domain.)
Assume that $\Xi\subset\mathbb{R}$ is nonempty and compact.

\item (Regularity of scalarized criterion function.)
Let $\bar{\mathcal{Z}}:=\mathrm{conv}(\Psi(\mathcal{Z}))\subset\mathbb{R}^{d_\Psi}$, where $\mathcal{Z}$ is the compact set from (R5) and $\mathrm{conv}(\cdot)$ denotes the convex hull.

Assume the functional $\mathcal{H}:\Xi\times\mathbb{R}^d\times\mathbb{R}^{d_\Psi}\to\mathbb{R}$ is continuous in $(\xi,s,\bar s)$ on $\Xi\times\mathcal{Z}\times\bar{\mathcal{Z}}$.

In particular, since $\Xi\times\mathcal{Z}\times\bar{\mathcal{Z}}$ is compact, $\mathcal{H}$ is bounded on $\Xi\times\mathcal{Z}\times\bar{\mathcal{Z}}$.
\end{enumerate}
\end{assumption}

%\noindent\textbf{Convergence pipeline and role of Assumption~\ref{ass:rr_regularity}.}
%The convergence analysis proceeds by propagating NN approximation through successive layers of the control problem:
%\[
%\text{NN feedback}
%\rightarrow
%\text{state sequence}
%\rightarrow
%\text{performance and moment vectors}
%\rightarrow
%\text{scalarized criterion function $\mathcal{H}$ for fixed $\xi$}
%\rightarrow
%\text{scalarized objective $\mathcal{H}$ for fixed $\xi$}
%\rightarrow
%\text{value function $V$}
%%(\text{stability of} \sup_{\xi\in\Xi})
%\]
\subsection{Discussion of Assumption~\ref{ass:rr_regularity}}
\label{ssc:rr_regularity_discussion}
Assumption~\ref{ass:rr_regularity} ensures that convergence of the NN feedback controls propagates through the controlled recursion via continuity of the update maps, induces convergence of the finite-dimensional performance/moment vectors, and is preserved under the scalarized objective through continuity and compact-domain uniformity. Most conditions are local/compact-domain regularity statements, with compactness used mainly to obtain boundedness and uniform continuity needed for uniform laws of large numbers (ULLN).

\medskip
\noindent\textit{Regarding (R1)}.
The uniform bounds on $W(t_m^\pm;\mathcal{P},Y)$ are strong but often natural in applications with physical or contractual state constraints (e.g.\ storage with capacity limits, state processes with reflecting/truncating mechanisms, or state updates that enforce admissible ranges by construction). In settings where the state is not globally bounded, one can sometimes replace (R1) by straightforward localization arguments \cite{tsang2020deep, hure2021deep}, but this may complicate the continuity and uniform-convergence steps in subsequent proofs.

\medskip
\noindent\textit{Regarding (R2).}
Continuity is needed in all arguments of the update maps $U_q$ and $U_p$ that enter the recursion. In particular, the post-decision update \eqref{eq:rr_skeleton_evolution}
depends explicitly on the exogenous input $Y_{m+1}$. The continuity requirement for $U_p$ does not assume that $Y_m$ is bounded; instead, it ensures continuity at all finite realizations of $Y_{m+1}$. If, in a particular application, one has an a.s.\ bound $\|Y_m\|\le r_0$ for some finite $r_0>0$, then it suffices to assume continuity only on the set $\{y\in\mathbb{R}^{d_a}:\ \|y\|\le r_0\}$. Moreover, if $U_q$ or $U_p$ are piecewise-defined (e.g.\ due to regime switching), it suffices that the realized recursion is continuous in the current state and action variables on the relevant bounded regime.

\medskip
\noindent\textit{Regarding (R3).}
Existence of an optimizer $(\mathcal{P}^\ast,\xi^\ast)$ is a standard standing assumption in convergence results for value functions: it lets us compare the NN value against an attained optimum. If an optimizer fails to exist, many arguments can be adapted by working with an $\varepsilon$-optimal sequence $(\mathcal{P}^{\varepsilon},\xi^{\varepsilon})$ \cite{bertsekas1996stochastic}. We keep (R3) to avoid additional bookkeeping.
%
%for instance, by definition of the supremum one may select $(\mathcal{P}^{\varepsilon},\xi^{\varepsilon})$ such that $V(w_0,t_0^-;\xi^{\varepsilon},\mathcal{P}^{\varepsilon})\ge V(w_0,t_0^-)-\varepsilon$, and then take $\varepsilon\downarrow0$ to obtain a near-optimal maximizing sequence. Such $\varepsilon$-optimal policies/selections are standard in discrete-time stochastic control; see, e.g., \cite{bertsekas1978stochastic}.

\medskip
\noindent\textit{Regarding (R4).}
This condition permits discontinuous optimal feedback maps (e.g.\ threshold/bang--bang rules) provided the optimal state does not hit their discontinuity sets with positive probability at intervention times. It is strictly weaker than assuming continuity of $q^\ast(t_m^-,\cdot)$ and $p^\ast(t_m^+,\cdot)$ on $[w_{\min},w_{\max}]$. Technically, it is the minimal condition needed to apply extended continuous-mapping arguments in the decision and state convergence proofs (see Lemmas~\ref{lem: convergence of q amount}--\ref{lem: convergence of W}).

\medskip
\noindent\textit{Regarding (R5).}
A simple sufficient way to ensure the existence of a compact set $\mathcal{Z}$ in (R5) is to combine:
(i) the uniform state bounds $W(t_m^\pm)\in[w_{\min},w_{\max}]$ a.s.\ for all $m$ and all $\mathcal{P}\in\mathcal{G}$, and
(ii) the action constraints $q(t,w)\in \mathcal{Z}_q(w)\subset[q_{\min},q_{\max}]$ and $p(t,w)\in\mathcal{Z}_p$ as stated in \eqref{eq: Z_q}-\eqref{eq: Z_p}.
Because $S(\mathcal{P},Y)$ is finite-dimensional and built from finitely many state/action components at the intervention times, these bounds imply that every coordinate of $S(\mathcal{P},Y)$ is a.s.\ bounded uniformly over $\mathcal{P}\in\mathcal{G}$, hence $S(\mathcal{P},Y)$ takes values in a compact product set, which can be taken as $\mathcal{Z}$.

\medskip
\noindent\textit{Regarding (R6).}
Compactness of $\Xi$ is mainly a technical convenience: combined with continuity of $\mathcal{H}$ on $\Xi\times\mathcal{Z}\times\bar{\mathcal{Z}}$, it implies (i) boundedness of $\mathcal{H}$ on this domain and (ii) uniform continuity properties (e.g.\ uniform continuity in the $\bar s$-argument uniformly over $(\xi,s)\in\Xi\times\mathcal{Z}$) that simplify subsequent proofs.
In many standard auxiliary-variable representations, it is  without loss of optimality to take $\Xi$ compact, since boundedness of the relevant performance statistic typically allows restricting $\xi$ to a compact interval.

In the examples below, $X$ denotes the scalar performance statistic entering the auxiliary-variable representation (e.g.\ a coordinate of $S(\mathcal{P},Y)$, hence a.s.\ bounded once (R5) holds). For the Rockafellar--Uryasev representation of $\mathrm{CVaR}(X)$, any optimizer $\xi^\ast$ lies in $[x_{\min},x_{\max}]$, so one may take $\Xi=[x_{\min},x_{\max}]$. Likewise, for the auxiliary-variable representation of variance,
$\sup_{\xi\in\mathbb{R}}\mathbb{E}[-(X-\xi)^2]$, the optimizer is $\xi^\ast=\mathbb{E}[X]\in[x_{\min},x_{\max}]$, so the search may be restricted to a compact interval.
For bPoE-type representations with the constraint $\xi>D$, one can similarly restrict $\Xi$ to a compact interval $[\,\underline\xi,\,\overline\xi\,] \subset(D, \infty)$
 that contains an optimizer.

More generally, it suffices that for each $(s,\bar s)\in\mathcal{Z}\times\bar{\mathcal{Z}}$ the map $\xi\mapsto \mathcal{H}(\xi,s,\bar s)$ is coercive, in the sense that $\mathcal{H}(\xi,s,\bar s)\to -\infty$ as $|\xi|\to\infty$ uniformly over $(s,\bar s)$ on compact sets. In this case the maximizer over $\xi$ must lie in some bounded interval, and one may replace $\Xi$ by a compact interval containing all optimizers \cite{rockafellar1998variational}.

\medskip
\noindent{\textit{Regarding (R7).}
Recall the moment vector $\bar{S}(\mathcal{P}):=\mathbb{E}[\Psi(S(\mathcal{P},Y))]$ defined in \eqref{eq:rr_moment_vector}. Then $\bar{S}(\mathcal{P})\in\bar{\mathcal{Z}}$ for all $\mathcal{P}\in\mathcal{G}$, where $\bar{\mathcal{Z}}:=\mathrm{conv}(\Psi(\mathcal{Z}))$, so the third argument of $\mathcal{H}$ ranges over $\bar{\mathcal{Z}}$.
To see this, note that (R5) gives $S(\mathcal{P},Y)\in\mathcal{Z}$ a.s.\ for each $\mathcal{P}\in\mathcal{G}$, hence $\Psi(S(\mathcal{P},Y))\in\Psi(\mathcal{Z})$ a.s. Since $\Psi$ is continuous on the compact set $\mathcal{Z}$, the image $\Psi(\mathcal{Z})$ is compact. In finite dimensions, $\mathrm{conv}(\Psi(\mathcal{Z}))$ is also compact, and because it is convex and contains $\Psi(\mathcal{Z})$, it contains the expectation of any $\Psi(\mathcal{Z})$-valued random vector; in particular, $\bar{S}(\mathcal{P})=\mathbb{E}[\Psi(S(\mathcal{P},Y))]\in\bar{\mathcal{Z}}$.

With (R6)--(R7), continuity of $\mathcal{H}$ on the compact set $\Xi\times\mathcal{Z}\times\bar{\mathcal{Z}}$ yields boundedness and the uniform continuity properties used in the subsequent proofs involving ULLN.

\section{A neural network approach}
\label{sec:NN}
We approximate the two-step feedback control $\mathcal{P}=(q,p)$ by two interacting fully connected FNNs: a scalar network for the pre-decision action $q$ and a vector network for the post-decision action $p$.
The pointwise constraints \eqref{eq: Z_q}--\eqref{eq: Z_p} are enforced by customized output-layer maps, yielding an unconstrained optimization problem over the network parameters.

\subsection{Preliminaries}
We use standard fully connected FNNs (multilayer perceptrons), i.e.\ compositions of affine maps and nonlinear activations.

\begin{definition}[Feedforward neural network]
\label{defn: FNN}
Consider a fully connected multilayer FNN with $L$ hidden layers. Layers are indexed by
$l \in\{0,1,\ldots,L+1\}$, where $l=0$ is the input layer and $l=L+1$ is the output
layer. Let $\nu_{l}\in\mathbb{N}$ be the number of nodes in layer $l$. The function
$F:\mathbb{R}^{\nu_0}\to\mathbb{R}^{\nu_{L+1}}$ computed by the FNN is
\begin{equation}
\label{eq: FNN}
F
=
A_{L+1}\circ (\psi_{L}\circ A_{L})\circ \cdots \circ (\psi_{1}\circ A_{1}),
\end{equation}
where $A_l(x)=\beta_l x+b_l$ is affine with weights $\beta_l\in\mathbb{R}^{\nu_l\times \nu_{l-1}}$
and bias $b_l\in\mathbb{R}^{\nu_l}$, and $\psi_l:\mathbb{R}^{\nu_l}\to\mathbb{R}^{\nu_l}$ is the
activation function for $l=1,\ldots,L$. No activation is applied at the output layer.

We collect all trainable parameters (weights and biases) as
\mbox{$\theta=(\beta_l,b_l)_{l=1,\ldots,L+1}\in\mathbb{R}^{\eta}$,}
where $\eta$ denotes the total number of trainable parameters of the network.
\end{definition}

For notational simplicity, we take all hidden layers to have the same width $\nu$ and use the sigmoid activation
in each hidden layer. Let $\{\nu^{(n)}\}_{n\in\mathbb{N}}$ be a strictly increasing sequence with $\nu^{(n)}\to\infty$ as $n\to\infty$. For each $n$, let $\mathcal{Q}_{n}$ denote the class of FNNs
of the form \eqref{eq: FNN} with hidden-layer width $\nu^{(n)}$ (and fixed depth $L$). Then
$\mathcal{Q}_{n}\subseteq \mathcal{Q}_{n+1}$.

The next theorem states a universal approximation property in the probabilistic form convenient for our analysis.
Throughout, $\|\cdot\|$ denotes a fixed norm on $\Rbb^d$ (with ambient dimension clear from context).

\begin{theorem}[Universal approximation for a random input {\cite[Thm.~2.4 and Cor.~2.7]{hornik1989multilayer}}]
\label{thm: convergence of F}
Let $X$ be an $\mathbb{R}^{\nu_0}$-valued random variable and let
$f:\mathbb{R}^{\nu_0}\to\mathbb{R}^{d}$ be Borel measurable. Then there exists a sequence
$\{F_{n}\}_{n\in\mathbb{N}}$, where $F_{n}=F(\cdot;\theta_{n})\in \mathcal{Q}_{n}$, such that
for all $\varepsilon>0$,
\[
\lim_{n\to\infty}\mathbb{P}\!\left(\big\|F_{n}(X)-f(X)\big\|>\varepsilon\right)=0.
\]
\end{theorem}
To encode constraints, we apply (measurable) customized maps at the output layer. The next lemma records that convergence in probability is preserved under composition with mappings that are continuous at the limit point almost surely.
\begin{lemma}[Composition with (a.s.-continuous) activations]
\label{lem: NN composition}
Let $X$ be an $\mathbb{R}^{\nu_0}$-valued random variable. Let $f:\mathbb{R}^{\nu_0}\to\mathbb{R}^{d}$ and
$\psi:\mathbb{R}^{d}\to\mathbb{R}^{k}$ be measurable. Define the set of discontinuity points of $\psi$ by
\[
D_{\psi}
:=
\Big\{ y\in\mathbb{R}^{d}:\ y\mapsto \psi(y)\ \text{is discontinuous at } y \Big\},
\]
and assume that $\mathbb{P}\!\left(f(X)\in D_{\psi}\right)=0$.
Then there exists a sequence $\{F_{n}\}_{n\in\mathbb{N}}$ with
$F_{n}=F(\cdot;\theta_{n})\in \mathcal{Q}_{n}$ such that, for all $\varepsilon>0$,
\[
\lim_{n\to\infty}
\mathbb{P}\!\left(
\big\|\psi\!\left(F_{n}(X)\right)-\psi\!\left(f(X)\right)\big\|>\varepsilon
\right)=0.
\]
\end{lemma}
A proof of Lemma~\ref{lem: NN composition} is given in  Appendix~\ref{app: NN composition}.
%\begin{proof}[Proof of Lemma~\ref{lem: NN composition}]
%By Theorem~\ref{thm: convergence of F}, we have $F_{n}(X)\xrightarrow{P} f(X)$ as $n\to\infty$.
%By assumption, $\mathbb{P}\!\left(f(X)\in D_\psi\right)=0$, i.e.\ $\psi$ is continuous at $f(X)$ almost surely.
%Hence, by the (extended) continuous mapping theorem \cite{Billingsley1995,kallenberg2002},
%$\psi(F_{n}(X))\xrightarrow{P}\psi(f(X))$.
%\end{proof}

In practice, it is often convenient to use standard activations with open output ranges (e.g.\ sigmoid, softmax), which cannot represent boundary values exactly. The next lemma justifies approximation when the target range includes boundary points.
\begin{lemma}[Boundary approximation via open-range activations]
\label{lem: activation function}
Let $X$ be an $\mathbb{R}^{\nu_0}$-valued random variable and let $g:\mathbb{R}^{\nu_0}\to[a,b]^N$ be measurable.
Let $\psi:\mathbb{R}^N\to(a,b)^N$ be continuous and assume that $\psi$ admits a measurable right inverse
$\psi^{-1}_r:(a,b)^N\to\mathbb{R}^N$, i.e.\ $\psi\circ\psi^{-1}_r=\mathrm{Id}_{(a,b)^N}$.
Then, there exists a sequence $\{F_{n}\}_{n\in\mathbb{N}}$ with
$F_{n}=F(\cdot;\theta_{n})\in \mathcal{Q}_{n}$ such that, for all $\varepsilon>0$,
\[
\lim_{n\to\infty}\mathbb{P}\!\left(
\big\|\psi\!\left(F_{n}(X)\right)-g(X)\big\|>\varepsilon
\right)=0.
\]
\end{lemma}
For a proof of Lemma~\ref{lem: activation function}, see Appendix~\ref{app: activation function}

\subsection{Pre-decision network}
Let $\widetilde{z}(\cdot;\theta_{q,n}):\mathcal{D}_\phi\to\Rbb$ be an FNN of the form \eqref{eq: FNN}, where
$\theta_{q,n}$ denotes its trainable parameter vector. To enforce the state-dependent interval constraint
$q(t,w)\in \mathcal{Z}_q(w)$ from \eqref{eq: Z_q}, we use the customized output map
\begin{equation}
\label{eq: psi_q}
\psi_q(w,z)
:=
q_{\min}
+\mathrm{range}(w)\,\sigma(z),
\quad
\mathrm{range}(w):=\max\!\big(\min(q_{\max},w)-q_{\min},\,0\big),
\end{equation}
where $\sigma(z):=(1+e^{-z})^{-1}$ is the sigmoid function.
The resulting \mbox{pre-decision network is}
\begin{equation}
\label{eq: q NN}
\widehat{q}(t,w;\theta_{q,n})
:=
\psi_q\!\big(w,\widetilde{z}(t,w;\theta_{q,n})\big).
\end{equation}
By construction, $\widehat{q}(t,w;\theta_{q,n})\in \mathcal{Z}_q(w)$ for all $(t,w)\in\mathcal{D}_\phi$.

The next theorem records a probabilistic approximation result for $\widehat q$ at a random feature input.
\begin{theorem}
\label{thm: convergence of q}
Let $X=(t,w)$ be a $\mathcal{D}_\phi$-valued random variable. Then there exists a sequence of networks
$\{\widehat{q}(\cdot;\theta_{q,n}^\ast)\}_{n\in\mathbb{N}}$
such that, for all $\varepsilon>0$,
\[
\lim_{n\to\infty}
\mathbb{P}\!\left(\left|\widehat{q}(X;\theta_{q,n}^\ast)-q^\ast(X)\right|>\varepsilon\right)=0,
\]
where $q^\ast$ is any admissible target map (in particular, the optimal pre-decision control).
\end{theorem}
A full detail proof of Theorem~\ref{thm: convergence of q} is given in Appendix~\ref{app: convergence of q}.

\subsection{Post-decision network}
Let $\widetilde{p}(\cdot;\theta_{p,n}):\mathcal{D}_\phi\to\Rbb^{d_a}$ be an FNN of the form \eqref{eq: FNN}.
To enforce the simplex constraint \eqref{eq: Z_p}, we apply the softmax map $\psi_p=(\psi_p^i)_{i=1}^{d_a}$,
defined componentwise by
\begin{equation}
\label{eq: psi_p}
\psi_p^{i}(z)
:=
\frac{e^{z_i}}{\sum_{j=1}^{d_a}e^{z_j}},
\qquad z\in\mathbb{R}^{d_a},\ i=1,\ldots,d_a.
\end{equation}
The resulting post-decision network is
\begin{equation}
\label{eq: p NN}
\widehat{p}(t,w;\theta_{p,n})
:=
\psi_p\!\left(\widetilde{p}(t,w;\theta_{p,n})\right).
\end{equation}
By construction, $\widehat p(t,w;\theta_{p,n})\in\mathcal{Z}_p$ for all $(t,w)\in\mathcal{D}_\phi$.

\begin{theorem}
\label{thm: convergence of p}
Let $X$ be a $\mathcal{D}_\phi$-valued random variable. Then there exists a sequence of networks
$\{\widehat{p}(\cdot;\theta_{p,n}^\ast)\}_{n\in\mathbb{N}}$
such that, for all $\varepsilon>0$,
\[
\lim_{n\to\infty}
\mathbb{P}\!\left(\big\|\widehat{p}(X;\theta_{p,n}^\ast)-p^\ast(X)\big\|>\varepsilon\right)=0,
\]
where $p^\ast$ is any admissible target map (in particular, \mbox{the optimal post-decision control).}
\end{theorem}
A full detail proof of Theorem~\ref{thm: convergence of p} is given in Appendix~\ref{app: convergence of p}.

\subsection{NN-parametrized objective}
For a fixed architecture index $n$, corresponding to hidden-layer width $\nu^{(n)}$, we introduce the
NN parametrization of control pairs
\begin{equation}
\label{eq:NNpara_control}
\widehat{\mathcal{P}}(\Theta_{n})
:=
\left(\widehat{q}(\cdot;\theta_{q,n}),\widehat{p}(\cdot;\theta_{p,n})\right),
\quad \text{where} \quad
\Theta_{n}:=(\theta_{q,n},\theta_{p,n})\in\mathbb{R}^{\vartheta^{(n)}}.
\end{equation}
Here, $\widehat{q}(\cdot)$ and $\widehat{p}(\cdot)$ are defined in \eqref{eq: q NN} and \eqref{eq: p NN}, respectively. The total number of parameters
is $\vartheta^{(n)}:=\eta_q^{(n)}+\eta_p^{(n)}$,
where $\eta_q^{(n)}$ and $\eta_p^{(n)}$ denote the numbers of trainable parameters in the withdrawal and allocation networks, respectively.
Let $\widehat{\mathcal{G}}_{n}$ be the resulting class of
NN control pairs, given by
\begin{equation}
\label{eq: G_hat NN}
\widehat{\mathcal{G}}_{n}
:=
\left\{
\widehat{\mathcal{P}}(\Theta_{n}) : \Theta_{n}\in\mathbb{R}^{\vartheta^{(n)}}
\right\}.
\end{equation}

Given $\Theta_{n}$ and $Y$, let $\{W(t_m^-;\Theta_{n},Y),W(t_m^+;\Theta_{n},Y)\}_{m=0}^M$
be the state sequence induced by applying $\widehat{\mathcal{P}}(\Theta_{n})$ in the recursion
\eqref{eq:rr_skeleton_updates}--\eqref{eq:rr_skeleton_evolution}, i.e.\ with $q$ and $p$ replaced by $\widehat q$ and $\widehat p$.
Define the corresponding performance vector
\[
S(\Theta_{n},Y)
:= S(\widehat{\mathcal{P}}(\Theta_{n}),Y)
\]
as in \eqref{eq:rr_perf_vector}. For moment dependence, define also
\begin{equation}
\label{eq:barZ}
\bar{S}(\Theta_{n})
:=
\mathbb{E}^{w_0,t_0^-}_{\widehat{\mathcal{P}}(\Theta_{n})}\!\big[\,\Psi(S(\Theta_{n},Y))\,\big]
\in\mathbb{R}^{d_\Psi},
\end{equation}
as in \eqref{eq:rr_moment_vector} (and if $d_\Psi=0$, suppress $\bar{S}$ and $\Psi$).

Replacing $\mathcal{P}$ by $\widehat{\mathcal{P}}(\Theta_{n})$ in \eqref{eq:rr_objective} yields the NN objective
\begin{equation}
\label{eq:rr_objective_NN}
V_{NN}(w_0,t_0^-;\xi,\Theta_{n})
:=
\mathbb{E}^{w_0,t_0^-}_{\widehat{\mathcal{P}}(\Theta_{n})}\!\big[
\mathcal{H}\big(\xi,S(\Theta_{n},Y),\bar{S}(\Theta_{n})\big)
\big],
\end{equation}
and the associated NN value function
\begin{equation}
\label{eq:rr_value_NN}
V_{NN}(w_0,t_0^-)
:=
\sup_{\Theta_{n}\in\mathbb{R}^{\vartheta^{(n)}},\ \xi\in\Xi}
V_{NN}(w_0,t_0^-;\xi,\Theta_{n}).
\end{equation}
Since the constraints are enforced by the output maps \eqref{eq: psi_q} and \eqref{eq: psi_p}, the optimization in
\eqref{eq:rr_value_NN} is unconstrained over the parameter space.

%\subsection{Sample-based objective}
\subsection{Empirical objective}
In computation, expectations are approximated by sample averaging over i.i.d.\ realizations of the exogenous input process.
Let $Y^{(1)},\ldots,Y^{(K)}$ be i.i.d.\ copies of $Y$,
where, for each $k \in \{ 1, \ldots, K\}$,
$Y^{(k)} = \{Y_m^{(k)}\}_{m=1}^{M}$, is a full exogenous path.
We define the dataset
\begin{equation}
\label{eq: Y data}
\mathcal{Y}_K=\{Y^{(k)}:\ k=1,\ldots,K\}.
\end{equation}
For each $k$, let $S^{(k)}(\Theta_{n}) := S(\Theta_{n},Y^{(k)})$ denote the realized performance vector. We estimate $\bar{S}$ by the sample average
\begin{equation}
\label{eq:widehat_barZ}
\widehat{\bar{S}}_K(\Theta_{n})
:=
\frac{1}{K}\sum_{k=1}^K \Psi\!\left(S^{(k)}(\Theta_{n})\right)
\in\mathbb{R}^{d_\Psi}.
\end{equation}
The empirical objective is
\begin{align}
\widehat{V}_{NN}(w_0,t_0^-;\xi,\Theta_{n},\mathcal{Y}_K)
&:=
\frac{1}{K}\sum_{k=1}^{K}
\mathcal{H}\!\left(\xi,\ S^{(k)}(\Theta_{n}),\ \widehat{\bar{S}}_K(\Theta_{n})\right),
\nonumber
\\
&= \frac{1}{K}\sum_{k=1}^{K}
\Big(
\mathcal{R}\!\left(S^{(k)}(\Theta_{n})\right)
+\gamma\,\mathcal{L}\!\left(\xi,\ S^{(k)}(\Theta_{n}),\ \widehat{\bar{S}}_K(\Theta_{n})\right)
\Big),
\label{eq:rr_objective_sample}
\end{align}
and the corresponding empirical value function is
\begin{equation}
\label{eq:rr_value_sample}
\widehat{V}_{NN}(w_0,t_0^-;\mathcal{Y}_K)
:=
\sup_{\Theta_{n}\in\mathbb{R}^{\vartheta^{(n)}},\ \xi\in\Xi}
\widehat{V}_{NN}(w_0,t_0^-;\xi,\Theta_{n},\mathcal{Y}_K).
\end{equation}

\begin{remark}[Training algorithm]
\label{rm: NN training algorithm}
In the convergence analysis below, we assume that the chosen training algorithm attains a global optimizer of
\eqref{eq:rr_value_sample}. In practice, gradient-based methods may converge only to stationary points; we adopt this standing assumption to focus on approximation and estimation errors rather than optimization error.
\end{remark}

\section{Convergence analysis}
\label{sec:conv_analysis}
This section analyzes how NN parametrization of the optimal two-step feedback control propagate through the controlled recursion and into the generalized risk--reward objective. 
%The analysis separates into:
%(i) \emph{approximation error} (restricting to an NN control class with increasing width), and
%(ii) \emph{estimation error} (replacing population expectations by sample averages).
Throughout, we work under Assumption~\ref{ass:rr_regularity}. Let $(\mathcal{P}^\ast,\xi^\ast)\in\mathcal{G}\times\Xi$ be an optimizer as in
Assumption~\ref{ass:rr_regularity}~(R3), where $\mathcal{P}^\ast=(q^\ast,p^\ast)$.
For each architecture index $n$, let
\mbox{$\Theta_{n}^\ast=(\theta_{q,n}^\ast,\theta_{p,n}^\ast)$}
denote NN parameters producing a pre-decision action network $\widehat q(\cdot;\theta_{q,n}^\ast)$ and a post-decision action network $\widehat p(\cdot;\theta_{p,n}^\ast)$ that approximate the optimal controls $q^\ast$ and $p^\ast$ in probability at random feature inputs, as provided by Theorems~\ref{thm: convergence of q}--\ref{thm: convergence of p}.

We adopt the shorthand
\begin{equation}
\label{eq:w_shorthand}
W^\ast(t_m^\pm):=W(t_m^\pm;\mathcal{P}^\ast,Y),
\qquad
W^{(n)}(t_m^\pm):=W(t_m^\pm;\Theta_{n}^\ast,Y),
\end{equation}
for the optimal and NN-induced state variables at intervention times.

\subsection{Convergence of optimal actions}
\label{ssc:conv_withdrawals_general}
Lemmas~\ref{lem: convergence of q amount} and \ref{lem: convergence of p allocation}
below show that, at decision times, if the NN-induced state process converges to the optimal controlled state process, then the corresponding NN actions converge to the optimal actions.
\begin{lemma}
\label{lem: convergence of q amount}
Suppose Assumption~\ref{ass:rr_regularity} holds.
Let $\{\widehat{q}(\cdot;\theta_{q,n}^\ast)\}_{n\in\mathbb{N}}$ be the pre-decision action network sequence from
Theorem~\ref{thm: convergence of q}. Fix $m\in\{0,1,\ldots,M\}$.
If
\[
W^{(n)}(t_m^-)\xrightarrow{P}W^\ast(t_m^-)
\qquad\text{as }n\to\infty,
\]
then
\[
\widehat{q}\big(t_m^-,\,W^{(n)}(t_m^-);\theta_{q,n}^\ast\big)
\xrightarrow{P}
q^\ast\!\left(t_m^-,\,W^\ast(t_m^-)\right)
\qquad\text{as }n\to\infty .
\]
\end{lemma}

\begin{proof}[Proof of Lemma~\ref{lem: convergence of q amount}]
Fix $m$ and write $W_m^{(n)}:=W^{(n)}(t_m^-)$ and $W_m^\ast:=W^\ast(t_m^-)$.
By the triangle inequality,
\begin{equation}
\label{eq:fix1_split_q_general}
\begin{aligned}
\big|
\widehat{q}(t_m^-,W_m^{(n)};\theta_{q,n}^\ast)
-
q^\ast(t_m^-,W_m^\ast)
\big|
&\le
\big|
\widehat{q}(t_m^-,W_m^{(n)};\theta_{q,n}^\ast)
-
q^\ast(t_m^-,W_m^{(n)})
\big|
\\
&\qquad +
\big|
q^\ast(t_m^-,W_m^{(n)})
-
q^\ast(t_m^-,W_m^\ast)
\big|.
\end{aligned}
\end{equation}
By Assumption~\ref{ass:rr_regularity}~(R4),
$\mathbb{P}(W_m^\ast\in D_{q,m})=0$, i.e.\ $w\mapsto q^\ast(t_m^-,w)$ is continuous at $W_m^\ast$ almost surely.
Since $W_m^{(n)}\xrightarrow{P}W_m^\ast$ by assumption, the extended continuous mapping theorem yields
$q^\ast(t_m^-,W_m^{(n)})\xrightarrow{P}q^\ast(t_m^-,W_m^\ast)$.
So the second term on the rhs of \eqref{eq:fix1_split_q_general} vanishes in probability.

For the first term on the rhs, we handle the $n$-dependent input as follows.
Let $\mu_m$ denote the law of $W_m^\ast$ under $\mathbb{P}$. Since $\mu_m(D_{q,m})=0$, by outer regularity of $\mu_m$
there exist open sets $\{U_j\}_{j\in\mathbb{N}}$ such that
\[
D_{q,m}\subset U_j,
\qquad
\mu_m(U_j)<\frac{1}{j},
\qquad
\mu_m(\partial U_j)=0.
\]
Set $K_j:=[w_{\min},w_{\max}]\setminus U_j$, which is compact. By construction $K_j\cap D_{q,m}=\emptyset$,
hence $w\mapsto q^\ast(t_m^-,w)$ is continuous on $K_j$.

Fix $j\in\mathbb{N}$. By a deterministic universal approximation theorem in the
\mbox{$\sup$ norm \cite{cybenko1989approximation},}
we may choose the approximation sequence in Theorem~\ref{thm: convergence of q} so that
\begin{equation}
\label{eq:fix1_uniform_on_Kj_general}
\sup_{w\in K_j}
\big|
\widehat{q}(t_m^-,w;\theta_{q,n}^\ast)-q^\ast(t_m^-,w)
\big|
\;\longrightarrow\;0
\qquad\text{as }n\to\infty.
\end{equation}
In particular, for any $\varepsilon>0$ there exists $n(\varepsilon,j)$ such that for all $n\ge n(\varepsilon,j)$,
\[
\sup_{w\in K_j}
\big|
\widehat{q}(t_m^-,w;\theta_{q,n}^\ast)-q^\ast(t_m^-,w)
\big|
\le \varepsilon.
\]
Hence, for such $n$, using that $W_m^{(n)}\in[w_{\min},w_{\max}]$ a.s.\ by Assumption~\ref{ass:rr_regularity}~(R1),
\[
\mathbb{P}\!\left(
\big|
\widehat{q}(t_m^-,W_m^{(n)};\theta_{q,n}^\ast)-q^\ast(t_m^-,W_m^{(n)})
\big|>\varepsilon
\right)
\le
\mathbb{P}\!\left(W_m^{(n)}\in U_j\right).
\]
Since $W_m^{(n)}\xrightarrow{P}W_m^\ast$ implies $W_m^{(n)}\Rightarrow W_m^\ast$ and $\mu_m(\partial U_j)=0$,
the Portmanteau theorem \cite{Billingsley1999}[Sec.\ 2] yields
\[
\mathbb{P}(W_m^{(n)}\in U_j)\longrightarrow \mu_m(U_j)<\frac{1}{j}.
\]
Letting $n\to\infty$ and then $j\to\infty$ shows that the first term in \eqref{eq:fix1_split_q_general} also vanishes in probability.
Combining both terms yields the claim.
\end{proof}

\begin{lemma}
\label{lem: convergence of p allocation}
Suppose Assumption~\ref{ass:rr_regularity} holds.
Let $\{\widehat{p}(\cdot;\theta_{p,n}^\ast)\}_{n\in\mathbb{N}}$
be the post-decision action network sequence identified in Theorem~\ref{thm: convergence of p}.
Fix $m\in\{0,1,\ldots,M-1\}$.
If
\[
W^{(n)}(t_m^+)\xrightarrow{P}W^\ast(t_m^+)
\qquad\text{as }n\to\infty,
\]
then
\[
\widehat{p}\!\left(t_m^+,\,W^{(n)}(t_m^+);\theta_{p,n}^\ast\right)
\xrightarrow{P}
p^\ast\!\left(t_m^+,\,W^\ast(t_m^+)\right)
\qquad\text{as }n\to\infty .
\]
\end{lemma}
\begin{proof}[Proof of Lemma~\ref{lem: convergence of p allocation}]
The proof follows the same argument as Lemma~\ref{lem: convergence of q amount}, with $D_{q,m}$ replaced by $D_{p,m}$ and $|\cdot|$ replaced by $\|\cdot\|$.
\end{proof}

\subsection{Convergence of the controlled state sequence}
\label{ssc:conv_state_general}
We now show that, when the optimal feedback controls $(q^\ast,p^\ast)$ are approximated by the NNs from Theorems~\ref{thm: convergence of q}--\ref{thm: convergence of p}, the induced controlled state sequence converges in probability to the optimal state sequence at all intervention times.
\begin{lemma}
\label{lem: convergence of W}
Suppose Assumption~\ref{ass:rr_regularity} holds.
Let
$\{\widehat{q}(\cdot;\theta_{q,n}^\ast)\}_{n\in\mathbb{N}}$ and
$\{\widehat{p}(\cdot;\theta_{p,n}^\ast)\}_{n\in\mathbb{N}}$
be the NN sequences identified in Theorems~\ref{thm: convergence of q} and~\ref{thm: convergence of p}.
Recalling the notational convention \eqref{eq:w_shorthand}, the NN-induced controlled state satisfies
\[
W^{(n)}(t_m^\pm)
~\xrightarrow{P}~
W^\ast(t_m^\pm),
\quad \text{as } n\to\infty,
\qquad m=0,\ldots,M,
\]
where $t_M^+=T^+$ and $W(T^+;\cdot,Y)$ denotes the terminal post-decision state.
\end{lemma}
%{\myblue{The proof is based on ...
%A full detail proof of Theorem~\ref{thm: convergence of q} is given in
%Appendix~\ref{app: convergence of q}.}}
\begin{proof}[Proof of Lemma~\ref{lem: convergence of W}]
We prove the claim by induction on $m=0,1,\ldots,M-1$.

\medskip
\noindent\emph{Base case ($m=0$).}
This is trivial. By construction,
$W^{(n)}(t_0^-)=W(t_0^-;\Theta_{n}^\ast,Y)=w_0
= W(t_0^-;\mathcal{P}^\ast,Y)=W^\ast(t_0^-)$.

\medskip
\noindent\emph{Induction step.}
Fix $m\in\{0,\ldots,M-1\}$ and assume that
\[
W^{(n)}(t_m^-)\xrightarrow{P} W^\ast(t_m^-),
\quad \text{as } n\to\infty.
\]
By Lemma~\ref{lem: convergence of q amount},
\[
\widehat{q}\!\left(t_m^-,W^{(n)}(t_m^-);\theta_{q,n}^\ast\right)
~\xrightarrow{P}~
q^\ast\!\left(t_m^-,W^\ast(t_m^-)\right).
\]
Using the pre-decision update \eqref{eq:rr_skeleton_updates},
\begin{align*}
W^{(n)}(t_m^+)
&=
U_q\Big(t_m,\;W^{(n)}(t_m^-),\;\widehat{q}\!\left(t_m^-,W^{(n)}(t_m^-);\theta_{q,n}^\ast\right)\Big),
\\
W^\ast(t_m^+)
&=
U_q\Big(t_m,\;W^\ast(t_m^-),\;q^\ast\!\left(t_m^-,W^\ast(t_m^-)\right)\Big).
\end{align*}
By Assumption~\ref{ass:rr_regularity}~(R2), $U_q$ is continuous on the relevant bounded domain, hence the continuous mapping theorem yields
\[
W^{(n)}(t_m^+)\xrightarrow{P} W^\ast(t_m^+).
\]
Next, by Lemma~\ref{lem: convergence of p allocation},
\[
\widehat{p}\!\left(t_m^+,W^{(n)}(t_m^+);\theta_{p,n}^\ast\right)
~\xrightarrow{P}~
p^\ast\!\left(t_m^+,W^\ast(t_m^+)\right).
\]
Using the post-decision evolution update \eqref{eq:rr_skeleton_evolution},
\begin{align*}
W^{(n)}(t_{m+1}^-)
&=
U_p\Big(t_m,\;W^{(n)}(t_m^+),\;\widehat{p}\!\left(t_m^+,W^{(n)}(t_m^+);\theta_{p,n}^\ast\right),\;Y_{m+1}\Big),
\\
W^\ast(t_{m+1}^-)
&=
U_p\Big(t_m,\;W^\ast(t_m^+),\;p^\ast\!\left(t_m^+,W^\ast(t_m^+)\right),\;Y_{m+1}\Big).
\end{align*}
Again by continuity of $U_p$ on bounded sets  Assumption~\ref{ass:rr_regularity}~(R2), the continuous mapping theorem (applied jointly with $Y_{m+1}$, which does not depend on $n$) implies
\[
W^{(n)}(t_{m+1}^-)\xrightarrow{P} W^\ast(t_{m+1}^-).
\]
This completes the induction. Finally, applying Lemma~\ref{lem: convergence of q amount} at $m=M$ and the update \eqref{eq:rr_skeleton_updates} yields
\[
W^{(n)}(T^+)=W^{(n)}(t_M^+)\xrightarrow{P} W^\ast(t_M^+)=W^\ast(T^+),
\]
which completes the proof.
\end{proof}

\subsection{Convergence of performance vector and objective function}
\label{ssc:conv_objective_general}
We next show that, for any fixed auxiliary variable value $\xi$,
the scalarized objective evaluated under the NN approximations of the optimal control converges to the objective under the optimal control. This uses:
(i) convergence of the NN-induced controlled recursion (Lemma~\ref{lem: convergence of W}),
(ii) finiteness of the performance vector $S(\mathcal{P},Y)$,
(iii) continuity/boundedness of the scalarized criterion function $\mathcal{H}$ on the relevant bounded domain (Assumption~\ref{ass:rr_regularity}~(R7)).

\begin{lemma}
\label{lem: convergence of objective}
Suppose Assumption~\ref{ass:rr_regularity} holds.
Let
$\{\widehat{q}(\cdot;\theta_{q,n}^\ast)\}_{n\in\mathbb{N}}$ and
$\{\widehat{p}(\cdot;\theta_{p,n}^\ast)\}_{n\in\mathbb{N}}$
be the NN sequences identified in Theorems~\ref{thm: convergence of q} and~\ref{thm: convergence of p}.
Then, for any fixed $\xi\in\Xi$,
\[
\lim_{n\to\infty}
V_{NN}(w_0,t_0^-;\xi,\Theta_{n}^\ast)
=
V(w_0,t_0^-;\xi,\mathcal{P}^\ast),
\]
where $V$ is the objective \eqref{eq:rr_objective} and $V_{NN}$ is its NN-parametrized counterpart \eqref{eq:rr_objective_NN}.
\end{lemma}
For a proof of Lemma~\ref{lem: convergence of objective}, 
see Appendix~\ref{app: convergence of objective}.

\subsection{Convergence of value function}
\label{ssc:conv_value_general}
For each architecture index $n\in\mathbb{N}$, we make the dependence
of the NN-parametrized value function on the network class explicit by writing
\begin{equation}
\label{eq: VNN_nnu_def}
V_{NN}^{(n)}(w_0,t_0^-)
:=
\sup_{\Theta_{n}\in\mathbb{R}^{\vartheta^{(n)}},\ \xi\in\Xi}
V_{NN}(w_0,t_0^-;\xi,\Theta_{n}),
\end{equation}
where $V_{NN}(w_0,t_0^-;\xi,\Theta_{n})$ is defined in \eqref{eq:rr_objective_NN}.
We recall its empirical counterpart (based on an i.i.d.\ dataset $\mathcal{Y}_K=\{Y^{(k)}\}_{k=1}^K$) given in \eqref{eq:rr_value_sample}
\[
\widehat{V}_{NN}(w_0,t_0^-;\mathcal{Y}_K)
:=
\sup_{\Theta_{n}\in\mathbb{R}^{\vartheta^{(n)}},\ \xi\in\Xi}
\widehat{V}_{NN}(w_0,t_0^-;\xi,\Theta_{n},\mathcal{Y}_K),
\]
where $\widehat{V}_{NN}(w_0,t_0^-;\xi,\Theta_{n},\mathcal{Y}_K)$ is defined in \eqref{eq:rr_objective_sample}.

The total error admits the decomposition
\begin{equation}
\label{eq: value_error_decomp}
\begin{aligned}
\big|\widehat{V}_{NN}(w_0,t_0^-;\mathcal{Y}_K)-V(w_0,t_0^-)\big|
&\le
\big|\widehat{V}_{NN}(w_0,t_0^-;\mathcal{Y}_K)-V_{NN}^{(n)}(w_0,t_0^-)\big|
\\
&\qquad+
\big|V_{NN}^{(n)}(w_0,t_0^-)-V(w_0,t_0^-)\big|.
\end{aligned}
\end{equation}
The first term is the \emph{estimation error} (sample average vs.\ expectation), and the second term is the \emph{approximation error} (restriction to the NN control class).

\subsubsection{Vanishing approximation error}

\begin{theorem}
\label{thm: NN first approx}
Under Assumption~\ref{ass:rr_regularity}, we have
\[
\lim_{n\to\infty}
\Big|V_{NN}^{(n)}(w_0,t_0^-)-V(w_0,t_0^-)\Big|
=0.
\]
\end{theorem}

\begin{proof}[Proof of Theorem~\ref{thm: NN first approx}]
For each $n$, the constraint-enforcing output maps ensure that every NN control pair is admissible, hence the NN control class satisfies
$\widehat{\mathcal{G}}_{n}\subseteq \mathcal{G}$.
Therefore, $V_{NN}^{(n)}(w_0,t_0^-)=\ldots$
\begin{align*}
\ldots=
\sup_{\Theta_{n},\ \xi\in\Xi}
V(w_0,t_0^-;\xi,\widehat{\mathcal{P}}(\Theta_{n}))
\le
\sup_{\mathcal{P}\in\mathcal{G},\ \xi\in\Xi}
V(w_0,t_0^-;\xi,\mathcal{P})
=
V(w_0,t_0^-),
\end{align*}
so $\limsup_{n\to\infty}V_{NN}^{(n)}\le V$.

Let $(\mathcal{P}^\ast,\xi^\ast)$ be an optimizer from Assumption~\ref{ass:rr_regularity}~(R3), and let $\Theta_{n}^\ast$ be the NN parameter sequence identified above.
Since $V_{NN}^{(n)}$ is a supremum over $(\Theta_{n},\xi)$,
\[
V_{NN}^{(n)}(w_0,t_0^-)
\ge
V_{NN}(w_0,t_0^-;\xi^\ast,\Theta_{n}^\ast).
\]
By Lemma~\ref{lem: convergence of objective} (with $\xi=\xi^\ast$),
\[
V_{NN}(w_0,t_0^-;\xi^\ast,\Theta_{n}^\ast)
\to
V(w_0,t_0^-;\xi^\ast,\mathcal{P}^\ast)
=
V(w_0,t_0^-),
\]
hence $\liminf_{n\to\infty}V_{NN}^{(n)}\ge V$. Combining limsup and liminf yields the claim.
\end{proof}

\subsubsection{Vanishing estimation error for fixed architecture}
Fix $n$ and let $Y^{(1)},\ldots,Y^{(K)}$ be i.i.d.\ copies of $Y$, with dataset
$\mathcal{Y}_K=\{Y^{(k)}:\ k=1,\ldots,K\}$. The next lemma is a uniform law of large numbers (ULLN) for the empirical objective, including the case where the objective depends on the population moment $\bar S(\Theta_{n})$ through the sample-average estimator $\widehat{\bar S}_K(\Theta_{n})$.

\begin{lemma}[Uniform law of large numbers for the  NN objective]
\label{lem: NN ULLN}
Suppose Assumption~\ref{ass:rr_regularity} holds.
Fix $n\in\mathbb{N}$.
Then
\[
\sup_{\Theta_{n}\in\mathbb{R}^{\vartheta^{(n)}},\ \xi\in\Xi}
\big|
\widehat{V}_{NN}(w_0,t_0^-;\xi,\Theta_{n},\mathcal{Y}_K)
-
V_{NN}(w_0,t_0^-;\xi,\Theta_{n})
\big|
~\xrightarrow{P}~
0,
\quad \text{as } K\to\infty.
\]
\end{lemma}
For a proof of Lemma~\ref{lem: NN ULLN}, see Appendix~\ref{app: NN ULLN}.

\begin{theorem}
\label{thm: NN second approx v1}
Fix $n\in\mathbb{N}$.
Under the conditions of Lemma~\ref{lem: NN ULLN},
\[
\Big|
\widehat{V}_{NN}(w_0,t_0^-;\mathcal{Y}_K)
-
V_{NN}^{(n)}(w_0,t_0^-)
\Big|
~\xrightarrow{P}~
0,
\qquad \text{as } K\to\infty.
\]
\end{theorem}

\begin{proof}[Proof of Theorem~\ref{thm: NN second approx v1}]
Using the inequality $\big|\sup f-\sup g\big|\le \sup|f-g|$, we have
\begin{align*}
\big|
\widehat{V}_{NN}(w_0,t_0^-;\mathcal{Y}_K)
&-
V_{NN}^{(n)}(w_0,t_0^-)
\big|
\\
&\quad \le
\sup_{\Theta_{n},\ \xi\in\Xi}
\big|
\widehat{V}_{NN}(w_0,t_0^-;\xi,\Theta_{n},\mathcal{Y}_K)
-
V_{NN}(w_0,t_0^-;\xi,\Theta_{n})
\big|.
\end{align*}
The right-hand side converges to $0$ in probability by Lemma~\ref{lem: NN ULLN}.
\end{proof}

\subsubsection{Consistency of the two-step approximation}
We now state the main convergence result in the next theorem.
\begin{theorem}
\label{thm: NN second approx}
Suppose Assumption~\ref{ass:rr_regularity} holds and Lemma~\ref{lem: NN ULLN} holds for each fixed~$n$.
Then the empirical NN value function is consistent for the true value function in the following sense:
for every $\varepsilon>0$ and $\delta>0$, there exists $\bar n\in\mathbb{N}$ such that for every $n\ge \bar n$
there exists $K_0(n)\in\mathbb{N}$ with
\[
\mathbb{P}\!\left(
\Big|\widehat{V}_{NN}(w_0,t_0^-;\mathcal{Y}_K)-V(w_0,t_0^-)\Big|>\varepsilon
\right)
<\delta,
\qquad \forall\,K\ge K_0(n).
\]
\end{theorem}

\begin{proof}[Proof of Theorem~\ref{thm: NN second approx}]
Fix $\varepsilon>0$ and $\delta>0$. By Theorem~\ref{thm: NN first approx}, there exists $\bar n$ such that
\[
\big|V_{NN}^{(n)}(w_0,t_0^-)-V(w_0,t_0^-)\big|\le \varepsilon/2
\qquad \text{for all } n\ge \bar n.
\]
Fix any such $n$. By Theorem~\ref{thm: NN second approx v1}, there exists $K_0(n)$ such that for all $K\ge K_0(n)$,
\[
\mathbb{P}\!\left(
\Big|\widehat{V}_{NN}(w_0,t_0^-;\mathcal{Y}_K)-V_{NN}^{(n)}(w_0,t_0^-)\Big|>\varepsilon/2
\right)
<\delta.
\]
Using \eqref{eq: value_error_decomp} then yields
\[
\mathbb{P}\!\left(
\Big|\widehat{V}_{NN}(w_0,t_0^-;\mathcal{Y}_K)-V(w_0,t_0^-)\Big|>\varepsilon
\right)
<\delta,
\qquad \forall\,K\ge K_0(n),
\]
which proves the claim.
\end{proof}

\section{Numerical experiments}
\label{sec:num}
This section illustrates the convergence-in-probability results of Section~\ref{sec:conv_analysis} on a representative discrete-intervention risk--reward control problem.
Numerically, we study value-function convergence by
(i) increasing NN capacity (architecture index $n$, hence enriching the NN policy class $\widehat{\mathcal{G}}_{n}$ in \eqref{eq: G_hat NN}) at a fixed, large training sample size $K$, and (ii) increasing the training sample size $K$ at a fixed, sufficiently rich NN architecture $n$.

\subsection{A decumulation optimization problem}
We illustrate the full convergence pipeline using a stylized Defined Contribution (DC) retirement decumulation setting. This example is included only as an application-level illustration with interpretable interventions and controls; the analysis is not specific to finance.

\paragraph{Intervention times.}
We take yearly intervention times $t_m=m$, $m=0,\ldots,M$, with $M=30$ (so $\Delta t=1$ year). For this illustration, we abstract from mortality risk by working conditional on the retiree being alive through the fixed horizon $T$,
consistent with the practitioner ``plan-to-live, not to die'' convention (see, e.g.\ \cite{pfau2018overview}).

\paragraph{State and exogenous input process.}
To obtain a concrete, nontrivial input law while keeping the state dimension low enough to compute a high-accuracy reference value via a grid-based method, we set $d_a=2$ and consider two investable assets: a risky asset and a risk-free asset.

In this setting, the scalar controlled state is the inflation-adjusted portfolio balance (i.e.\ the ``wealth'' process) $\{W(t)\}_{0\le t\le T}$.
Between intervention times, the post-decision update \eqref{eq:rr_skeleton_evolution}
propagates the wealth process using a one-period gross return vector $Y_{m+1}$,
$m=0,\ldots,M-1$. Accordingly, $Y_{m+1}\in\Rbb^{2}$ represents the (inflation-adjusted) gross returns of the risky and risk-free assets over the interval $[t_m^+,t_{m+1}^-]$ and is observed at time $t_{m+1}^-$.

We model $Y_{m+1}$ using unit-investment (growth) processes.
Let $\{G^S_t\}_{t\in[0,T]}$ and $\{G^B_t\}_{t\in[0,T]}$ denote the inflation-adjusted values of one unit of currency invested in the risky and risk-free assets, respectively, normalized by $G^S_0=G^B_0=1$.
Between intervention times $t\in[t_m^+,t_{m+1}^-]$, the risky growth process follows a one-factor Kou jump--diffusion \cite{kou01}
and the risk-free growth process has deterministic accumulation:
\begin{equation}
\label{eq:dGS}
\frac{dG^S_t}{G^S_{t^-}}
=
(\mu-\lambda\kappa)\,dt+\sigma\,dZ_t
+
d\Big(\sum_{i=1}^{\pi_t}(\chi_i-1)\Big),
\qquad
dG^B_t=r_f\,G^B_t\,dt.
\end{equation}
Here, $\{Z_t\}$ is Brownian motion, $\{\pi_t\}$ is a Poisson process with constant intensity $\lambda>0$, the jump multipliers $\{\chi_i\}$ are i.i.d., $\kappa=\Ebb[\chi-1]$, and $\{Z_t\}$, $\{\pi_t\}$, and $\{\chi_i\}$ are mutually independent; in the Kou specification, $\log(\chi)$ has an asymmetric double-exponential law \cite{kou01}. We then define
\[
Y_{m+1}
=
\left(
G^S_{t_{m+1}^-}/G^S_{t_m^+},
\;
G^B_{t_{m+1}^-}/G^B_{t_m^+}
\right),
\qquad m=0,\ldots,M-1,
\]
and let $Y=\{Y_m\}_{m=1}^{M}$. Calibration of \eqref{eq:dGS} is discussed in Subsection~\ref{ssc:num_calibration}.
\paragraph{Admissible two-step controls.}
%In this setting, the scalar controlled state process is the inflation-adjusted account balance $\{W(t)\}_{0\le t\le T}$, referred to as the ``wealth process'' hereafter.
At each time $t_m$, the control is a two-step feedback policy $\mathcal{P}=(q,p)$: given the current wealth level, the retiree first (i) chooses a withdrawal $q(t_m^-,W(t_m^-;\cdot))$ at $t_m^-$ for $m=0,\ldots,M$, producing post-withdrawal wealth $W(t_m^+;\cdot)$, and then (ii) chooses portfolio weights $p(t_m^+,W(t_m^+;\cdot))$ at $t_m^+$ for $m=0,\ldots,M-1$;
\mbox{no allocation occurs at $t_M^+$.}
Withdrawals satisfy the state-dependent constraint $q(t,w)\in \mathcal{Z}_q(w)$ from \eqref{eq: Z_q}, and allocations satisfy the simplex constraint $p(t,w)\in\mathcal{Z}_p$ from \eqref{eq: Z_p}.

\paragraph{Controlled recursion.}
We use the recursion \eqref{eq:rr_skeleton_updates}--\eqref{eq:rr_skeleton_evolution} with the updates map
\begin{equation}
\label{eq:ew_cvar_map}
U_q(t,w,q)=w-q, \quad \text{and} \quad U_p(t,w,p,y)=
\begin{cases}
w\,\langle p,y\rangle, & \text{if } w>0,\\
w\,y_2, & \text{if } w\le 0,
\end{cases}
\end{equation}
where $y_2$ denotes the risk-free gross return component of $y = (y_1, y_2) \in\Rbb^2$.
Thus, with $q_m(\cdot):=q(t_m^-,W(t_m^-))$ and
$p_m(\cdot):=p(t_m^+,W(t_m^+))$, we have, \mbox{for $m=0,\ldots,M-1$},
\[
W(t_m^+)=W(t_m^-)-q_m(\cdot),
\qquad
W(t_{m+1}^-)=
\begin{cases}
W(t_m^+)\,\big\langle p_m(\cdot),\,Y_{m+1}\big\rangle, & \text{if } W(t_m^+)>0,\\[2pt]
W(t_m^+)\,Y_{m+1}^{(2)}, & \text{if } W(t_m^+)\le 0.
\end{cases}
%\qquad
\]
The treatment of $W(t_m^+)\le 0$ reflects liquidation at depletion, a practical constraint in DC decumulation plans: if withdrawals render wealth negative, allocation is halted and the wealth process thereafter follows deterministic debt accrual at the risk-free rate. The resulting recursion is piecewise-defined but continuous, consistent with the regime-switching structure in Subsection~\ref{ssc:rr_regularity_discussion} (R2).
%The treatment of  $W(t_m^+)\le 0$ reflects liquidation at depletion, a practical constraint in DC decumulation plans: if withdrawals render wealth negative, allocation is halted and the wealth process thereafter follows deterministic debt accrual at the risk-free rate. The resulting recursion is piecewise-defined but continuous, consistent with the regime-switching structure in Subsection~\ref{ssc:rr_regularity_discussion} (R2).

\noindent \paragraph{Performance vector.}
We take a performance vector to be
\[
S(\mathcal{P},Y)
:=
\Big(
\big(q(t_m^-,W(t_m^-;\mathcal{P},Y))\big)_{m=0}^{M},
\ \ W(T^+;\mathcal{P},Y)
\Big)\in\Rbb^{M+2}.
\]

\paragraph{Reward (expected cumulative withdrawal).}
Using \eqref{eq:rr_reward}, we set
\[
\mathcal{R}\big(S(\mathcal{P},Y)\big)
:=
\sum_{m=0}^{M} q\!\left(t_m^-,W(t_m^-;\mathcal{P},Y)\right).
\]

\paragraph{Risk (CVaR of terminal wealth).}
Using the template \eqref{eq:rr_risk} with $d_\Psi=0$,
we take $\alpha=0.05$ and represent $\mathrm{CVaR}_\alpha(W(T^+))$ via the Rockafellar--Uryasev form \cite{RT2000}:
\[
\mathcal{L}(\xi,s)
:=
\xi+\frac{1}{\alpha}\min\{\,s_{M+1}-\xi,\,0\},
\qquad (\xi,s)\in\Xi\times\Rbb^{M+2},
\]
where $s_{M+1}$ denotes the terminal-wealth component of $s= (s_0, \ldots, s_{M+1})$. Then
\[
\mathcal{J}_{\mathcal{L}}(w_0,t_0^-;\mathcal{P})
=
\sup_{\xi\in\Xi}\,
\mathbb{E}^{w_0,t_0^-}_{\mathcal{P}}\!\left[\mathcal{L}\big(\xi,S(\mathcal{P},Y)\big)\right]
=
\mathrm{CVaR}_{\alpha}\!\big(W(T^+)\big)
\quad\text{(gain-based convention)}.
\]

\paragraph{Scalarized objective.}
The scalarized criterion function \eqref{eq:rr_objective} becomes
\[
\mathcal{H}(\xi,s)
:=
\mathcal{R}(s)+\gamma\,\mathcal{L}(\xi,s), \qquad \gamma>0.
\]
The objective and value function are then
\begin{equation}
\label{eq:ew_cvar_new}
V(w_0,t_0^-;\xi,\mathcal{P})
=
\mathbb{E}^{w_0,t_0^-}_{\mathcal{P}}\!\big[
\mathcal{H}\big(\xi,S(\mathcal{P},Y)\big)
\big],
\quad
V(w_0,t_0^-)
\!=\!\!\!
\sup_{\mathcal{P}\in\mathcal{G},\ \xi\in\Xi} V(w_0,t_0^-;\xi,\mathcal{P}).
\end{equation}
We emphasize that the mean--CVaR objective is not dynamically separable in general; accordingly, the computed NN policy is interpreted as a pre-commitment solution (optimized at $t_0^-$ and then implemented via fixed feedback maps thereafter).
\subsection{Calibration}
\label{ssc:num_calibration}
We follow the calibration approach described in \cite{DM2016semi, PMVS2021c} using long-horizon Australian data.
The risky-asset inputs are based on monthly total returns of the ASX All Ordinaries Index (1935:01--2024:06), obtained from Bloomberg and supplemented with historical dividend-yield data where necessary \cite{mathews2019history}.
The risk-free rate is proxied by Australian 10-year government bond yields (Bloomberg), extended with historical sources to align horizons \cite{butlin2014statistical}.
All series are converted to real (inflation-adjusted) terms using CPI data from the Australian Bureau of Statistics.
The resulting annualized parameter values used in our experiments are reported in Table~\ref{tab:num_kou_params}.
\begin{table}[!hbt]
\vspace*{-0.4cm}
%\caption{Calibrated parameters for the exogenous-input model (annualized, inflation adjusted).}
\caption{Calibrated parameters (annualized, inflation adjusted).}
\label{tab:num_kou_params}
\vspace*{-0.2cm}
\centering{}
\begin{tabular}{c c c c c c c }
\hline
$\mu$ & $\sigma$ & $\lambda$ & $p_{up}$ & $\eta_{1}$ & $\eta_{2}$ & {{$r_f$}} \\
\hline
0.0774 & 0.1202 & 0.3243 & 0.3793 & 7.7209 & 5.9989 & 0.0126\\
\hline
\end{tabular}
\vspace*{-0.2cm}
\end{table}
Given these parameters, we generate i.i.d.\ scenario paths $Y^{(1)},\ldots,Y^{(K)}$ by Monte Carlo (MC) simulation of \eqref{eq:dGS}  and set $\mathcal{Y}_K=\{Y^{(k)}\}_{k=1}^K$ as in \eqref{eq: Y data}.

\subsection{Retirement scenario}
\label{ssc:num_scenario}
The retirement scenario is summarized in Table~\ref{tab:num_scenario}. The retiree selects annual withdrawals $q_m(\cdot)\in[q_{\min},q_{\max}]$, as in \eqref{eq: Z_q},  and portfolio weights satisfy the simplex constraint \eqref{eq: Z_p}.
Economically, $q_{\max}$ may be interpreted as a desired annual real spending level, while $q_{\min}$ is a contingency floor adopted to mitigate depletion risk; the state-dependent constraint \eqref{eq: Z_q} then permits reductions toward $q_{\min}$ in adverse wealth states while preserving higher withdrawals when the account remains well funded.
\begin{table}[h]
\vspace*{-0.4cm}
\caption{Retirement scenario. Monetary units are in thousands of real
dollar.}
\label{tab:num_scenario}
\vspace*{-0.2cm}
\centering
\begin{tabular}{l|c}
\hline
Retiree                 & 65-year-old Australian male\\
Investment horizon $T$ & $30$ years \\
Initial wealth $w_{0}$ & $1000$ \\
Intervention times & $t_m=m$, $m=0,1,\ldots,30$ \\
Annual withdrawal range & $\left[q_{\min}, q_{\max}\right] = \left[35, 60\right]$ \\
Withdrawal decision times & $m=0,1,\ldots,30$ (at $t_m^-$) \\
Allocation decision times & $m=0,1,\ldots,29$ (at $t_m^+$) \\
CVaR confidence level & $\alpha=0.05$ \\
Scalarization parameter & $\gamma=1.0$ \\
\hline
\end{tabular}
\vspace*{-0.2cm}
\end{table}

\subsection{Reference value}
\label{ssc:num_reference}
To assess numerical accuracy, we compare the NN-based empirical optimum
$\widehat V_{NN}(w_0,t_0^-;\mathcal{Y}_K)$ in \eqref{eq:rr_value_sample} against a
high-accuracy grid-based reference approximation $V_{\mathrm{ref}}$, computed for the same model \eqref{eq:dGS}
and risk--reward objective \eqref{eq:ew_cvar_new} using a provably convergent numerical integration method adapted from
\cite{dang2026multi} and specialized to include withdrawal controls.

For each fixed $\xi$, we solve the inner $(q,p)$ control problem by backward recursion on a truncated state domain.
One-period conditional expectations are evaluated by a monotone quadrature rule based on a nonnegative series
representation of the 1D Kou transition density (cf.\ \cite{dang2026multi}).
{\myblue{To handle intervention-time maximizations efficiently, we use the fact that admissible controls depend only on $(t,\text{total wealth})$: at each time step, optimal withdrawals and allocations are obtained by discrete search on a 1D wealth grid, rather than by introducing additional 2D control grids at every spatial node. This substantially reduces computational cost.}}
Finally, we maximize over $\xi\in\Xi$ by a 1D search.

The truncation choices are:
(i) a risky log-domain $y\in[\log(10^{2})-8,\ \log(10^{2})+8]$ with extended quadrature domain $y^\dagger\in[\log(10^{2})-16,\ \log(10^{2})+16]$;
(ii) a risk-free component truncated to $b\in[b_{\min},b_{\max}]=[-5\times 10^{5},\,5\times 10^{5}]$;
(iii) a log-wealth grid $w'=\log(w)\in[w'_{\min},w'_{\max}]
= [-10+\log(10^2),\,10+\log(10^2)]$ {\myblue{used for the 1D searches when $w>0$;
(for $w\le 0$ the recursion follows the liquidation/debt convention \eqref{eq:ew_cvar_map});}}
and (iv)  $\Xi=[-5\times 10^{5},\,5\times 10^{5}]$.

We report three refinement levels: $(N_y\times N_b)\in\{(512 \times512),(1024 \times1024),(2048\times2048)\}$.
{\myblue{At each refinement level, the intervention-time searches use a wealth grid with $N_w=4N_y$ nodes, and the outer $\xi$-search uses $N_\xi=N_b$ nodes.}}
The resulting approximate values at $(w_0,t_0^-)=(1000,0^-)$ are
$V_{(512\times512)} = 1600.08$, $V_{(1024 \times1024)} =1604.19$, and $V_{(2048\times2048)} =1605.22$. A convergence table 
is presented in Appendix~\ref{app:num_reference}.
We take the finest-grid value as
\[
V_{\mathrm{ref}}= 1605.22.
\]
Note that $|V_{\mathrm{ref}} - V_{(1024 \times1024)}|$ is about $0.06\%$ of $V_{\mathrm{ref}}$), which indicates that the discretization error is negligible at the scale of the NN errors reported in Subsection~\ref{ssc:num_convergence}.

\subsection{Training}
\label{ssc:num_training}
We parameterize the two-step feedback policy $\mathcal{P}=(q,p)$ by the NN class
$\widehat{\mathcal{G}}_n$ defined in \eqref{eq: G_hat NN}.
%In particular, we use the constraint-enforcing output maps \eqref{eq: psi_q} (interval constraint for withdrawals) and \eqref{eq: psi_p} (simplex constraint for allocations).
For a chose architecture index $n$ and dataset $\mathcal{Y}_K$, training maximizes the empirical objective \eqref{eq:rr_value_sample} jointly over NN parameters and the auxiliary variable $\xi$.
Each training run initializes $(\theta_{q,n},\theta_{p,n},\xi)$ randomly and applies Adam to maximize $\widehat V_{NN}(w_0,t_0^-;\xi,\Theta_n,\mathcal{Y}_K)$ using minibatches of scenario paths from $\mathcal{Y}_K$. We use separate learning rates for $(\theta_{q,n},\theta_{p,n})$ and for $\xi$, and a fixed iteration budget, consistent with the empirical objective studied in Section~\ref{sec:conv_analysis}.
Representative hyperparameters are given in Table~\ref{tab: NN hyper-parameters}.
\begin{table}[!hbt]
\vspace*{-0.4cm}
\caption{NN hyper-parameters}
\label{tab: NN hyper-parameters}
\vspace*{-0.2cm}
\centering
\begin{tabular}{l|l}
\hline
Iterations & $50{,}000$ \\
Minibatch size & $1{,}000$ \\
Initial Adam learning rate (NN parameters) & $0.05$ \\
Initial Adam learning rate (auxiliary $\xi$) & $0.04$ \\
Weight decay (L2 penalty) & $10^{-4}$ \\
\hline
\end{tabular}
\vspace*{-0.2cm}
\end{table}
In these experiments, we do not split $\mathcal{Y}_K$ into separate training and validation/test sets: for each run, $\mathcal{Y}_K$ is exactly the dataset defining the empirical objective \eqref{eq:rr_value_sample}, and we report the corresponding empirical optimum $\widehat V_{NN}(w_0,t_0^-;\mathcal{Y}_K)$. This matches the objects in the convergence in probability analysis (Theorem~\ref{thm: NN second approx}).
Out-of-sample robustness is discussed in Subsection~\ref{ssc:num_oos}.
\subsection{Convergence}
\label{ssc:num_convergence}
We empirically illustrate convergence in probability by repeating the full training-and-optimization procedure under independent randomization (scenario generation and random network's parameters initialization) and estimating tail probabilities of the error, in the spirit of Theorem~\ref{thm: NN second approx}.
Fix an experimental setting (architecture index $n$ and sample size $K$) and perform $N_{\mathrm{run}}$ independent runs, indexed by $j=1,\ldots,N_{\mathrm{run}}$.
In run $j$, we train the networks on a dataset $\mathcal{Y}_K^{(j)}$ and record the resulting empirical optimum
\[
\widehat V^{(j)}_{n,K}
\;:=\;
\widehat V_{NN}\!\left(w_0,t_0^-;\mathcal{Y}_K^{(j)}\right),
\]
%$\widehat V^{(j)}_{n,K} :=\widehat V_{NN}\!\left(w_0,t_0^-;\mathcal{Y}_K^{(j)}\right)$,
where $\widehat V_{NN}(w_0,t_0^-;\mathcal{Y}_K)$ is the maximized empirical objective in \eqref{eq:rr_value_sample}
(i.e.\ maximized over $\xi\in\Xi$ and NN parameters in $\Theta_n$)
for the policy class $\widehat{\mathcal{G}}_n$.
For a tolerance $\varepsilon>0$, we estimate the tail probability
$\mathbb{P}\!\left(\left|\widehat V_{n,K}-V_{\mathrm{ref}}\right|>\varepsilon\,|V_{\mathrm{ref}}|\right)$
by the empirical frequency
\begin{equation}
\label{eq:num_p_tail}
p_{\varepsilon}
:=
\frac{1}{N_{\mathrm{run}}}\sum_{j=1}^{N_{\mathrm{run}}}
\mathbf{1}\Big\{\big|\widehat V^{(j)}_{n,K}-V_{\mathrm{ref}}\big|>\varepsilon\,|V_{\mathrm{ref}}|\Big\}.
% \quad
%\text{where }
%\widehat V^{(j)}_{n,K} :=\widehat V_{NN}\!\left(w_0,t_0^-;\mathcal{Y}_K^{(j)}\right).
\end{equation}
%(When reporting tables below, $p_{\varepsilon}$ is computed separately for each row's setting.)
\subsubsection{Increasing NN capacity}
We first vary NN capacity while keeping the training sample size large, to reduce sampling noise and emphasize approximation effects.
In the notation of Section~\ref{sec:conv_analysis}, NN capacity is indexed by $n$; for each $n$ we take both networks in the two-step policy class $\widehat{\mathcal{G}}_n$ to have $L^{(n)}$ hidden layers and width $\nu^{(n)}$ (nodes per hidden layer).
We fix $K=2.56\times 10^5$ and consider
\[
(L^{(n)},\nu^{(n)})\in\{(1,2),\ (1,5),\ (2,5)\}.
\]
For each capacity level, we fix one dataset $\mathcal{Y}_K$ and run $N_{\mathrm{run}}=100$ independent trainings with different random initializations
(so the reported dispersion is conditional on this fixed scenario set).

Figure~\ref{fig:num_boxplot_capacity} summarizes the resulting $\widehat V^{(j)}_{n,K}$ across runs, and
Table~\ref{tab:num_capacity_conv} reports summary statistics and the empirical tail probabilities \eqref{eq:num_p_tail}. Key observations are:
\begin{itemize}[noitemsep, topsep=2pt, leftmargin=*]
\item As capacity increases, the distribution shifts toward $V_{\mathrm{ref}}$ (median/mean increase), consistent with improved approximation within richer classes (cf.\ Theorem~\ref{thm: NN first approx}).
\item For fixed $\varepsilon$, the estimated tail probabilities $p_{\varepsilon}$ decrease sharply with capacity (e.g.\ for $\varepsilon=1\%$ from $0.64$ to $0.02$), illustrating concentration of the empirical optimum around $V_{\mathrm{ref}}$ as approximation error decreases.
\end{itemize}
\begin{table}[!hbt]
%\vspace*{-0.5cm}
\caption{Empirical convergence as NN capacity increases (fixed $K=2.56\times 10^{5}$, $N_{\mathrm{run}}=100$ runs; $V_{\mathrm{ref}}=1605.22$); $p_{\varepsilon}$ is defined in \eqref{eq:num_p_tail}.}
\label{tab:num_capacity_conv}
\vspace*{-0.2cm}
\centering
\setlength{\tabcolsep}{6pt}
\begin{tabular}{c|c|ccccc}
\hline
$(L^{(n)},\nu^{(n)})$ & Mean (Std) & $p_{0.5\%}$ & $p_{1\%}$ & $p_{1.5\%}$ & $p_{2\%}$ & $p_{2.5\%}$ \\
%\hline
%$(1,2)$ & $1589.32\ (5.48)$ & $0.73$ & $0.55$ & $0.02$ & $0.00$ & $0.00$ \\
%\hline
%$(1,5)$ & $1593.68\ (6.07)$ & $0.62$ & $0.32$ & $0.00$ & $0.00$ & $0.00$ \\
%\hline
%$(2,5)$ & $1602.50\ (4.61)$ & $0.05$ & $0.01$ & $0.00$ & $0.00$ & $0.00$ \\
%\hline
\hline
$(1,2)$ & $1589.32\ (5.48)$ & $0.97$ & $0.64$ & $0.06$  & $0.00$ & $0.00$  \\
\hline
$(1,5)$ & $1593.68\ (6.07)$ & $0.77$ & $0.32$ & $0.00$  & $0.00$ & $0.00$  \\
\hline
$(2,5)$ & $1602.50\ (4.61)$ & $0.06$ & $0.02$ & $0.00$  & $0.00$ & $0.00$  \\
\hline
\end{tabular}
%\vspace*{-0.5cm}
\end{table}

\begin{figure}[htbp]
\centering
\begin{subfigure}[t]{0.48\linewidth}
\centering
\includegraphics[width=0.95\linewidth]{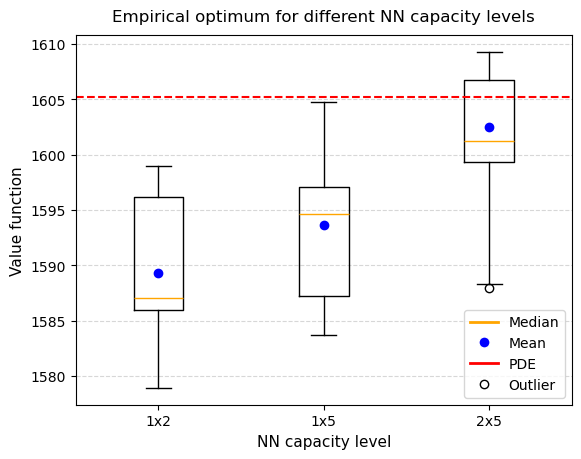}
\caption{Increasing NN capacity (fixed $K=2.56\times 10^5$).}
\label{fig:num_boxplot_capacity}
\end{subfigure}
\hfill
\begin{subfigure}[t]{0.48\linewidth}
\centering
\includegraphics[width=0.95\linewidth]{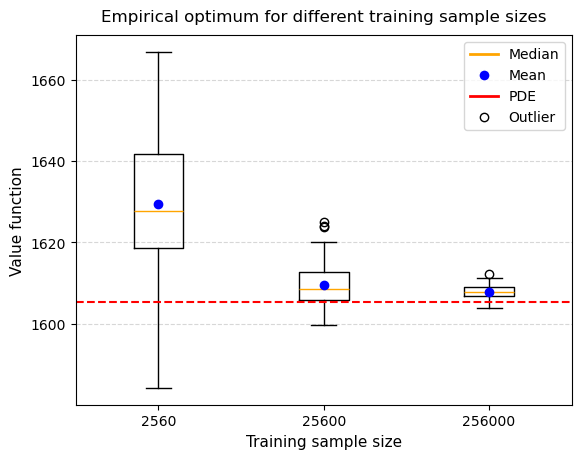}
\caption{Increasing training sample size (fixed capacity $(2,5)$).}
\label{fig:num_boxplot_sample}
\end{subfigure}
%\vspace*{-0.5cm}
\caption{Empirical optima $\widehat V^{(j)}_{n,K}$ across $N_{\mathrm{run}}=100$ runs. Boxes show the interquartile range (25\%--75\%) with median line; whiskers extend to $1.5\times\mathrm{IQR}$ and points beyond are plotted as outliers. The dashed line indicates the reference value $V_{\mathrm{ref}}=1605.22$.}
\label{fig:num_boxplots}
%\vspace*{-0.8cm}
\end{figure}
\subsubsection{Increasing training sample size}
Next, we fix a sufficiently rich architecture at $(L^{(n^\star)},\nu^{(n^\star)})=(2,5)$ and vary $K\in 2.56\times\{10^3,10^4,10^5\}$ to emphasize estimation effects. For each $K$, we perform $N_{\mathrm{run}}=100$ independent runs, each with a freshly simulated dataset $\mathcal{Y}_K^{(j)}$ and a fresh random initialization.

\begin{table}[!hbt]
\vspace*{-0.4cm}
\caption{Empirical convergence as training sample size $K$ increases  (fixed capacity $(L^{(n^\star)},\nu^{(n^\star)})=(2,5)$, $N_{\mathrm{run}}=100$ runs; $V_{\mathrm{ref}}=1605.22$); $p_{\varepsilon}$ is defined in \eqref{eq:num_p_tail}.}
\label{tab:num_sample_conv}
\vspace*{-0.2cm}
\centering
\setlength{\tabcolsep}{6pt}
\begin{tabular}{c|c|ccccc}
\hline
$K$ & Mean (Std) & $p_{0.5\%}$ & $p_{1\%}$ & $p_{1.5\%}$ & $p_{2\%}$ & $p_{2.5\%}$ \\
%\hline
%$2.56\times 10^{3}$ & $1629.41\ (18.27)$ & $0.90$ & $0.70$ & $0.49$ & $0.34$ & $0.21$ \\
%\hline
%$2.56\times 10^{4}$ & $1609.38\ (5.20)$ & $0.27$ & $0.03$ & $0.00$ & $0.00$ & $0.00$ \\
%\hline
%$2.56\times 10^{5}$ & $1607.87\ (1.58)$ & $0.01$ & $0.00$ & $0.00$ & $0.00$ & $0.00$ \\
%\hline
\hline
$2.56 \times 10^{3}$ & $1629.41\ (18.27)$ & $0.88$ & $0.67$ & $0.45$  & $0.32$ & $0.20$  \\
\hline
$2.56 \times 10^{4}$ & $1609.38\ (5.20)$ & $0.21$ & $0.03$ & $0.00$  & $0.00$ & $0.00$  \\
\hline
$2.56 \times 10^{5}$ & $1607.87\ (1.58)$ & $0.00$ & $0.00$ & $0.00$  & $0.00$ & $0.00$  \\
\hline
\end{tabular}
%\vspace*{-0.2cm}
\end{table}
Figure~\ref{fig:num_boxplot_sample} shows that the dispersion of $\widehat V^{(j)}_{n^\star,K}$ decreases substantially as $K$ increases.
Table~\ref{tab:num_sample_conv} reports the corresponding tail probabilities \eqref{eq:num_p_tail}. \mbox{Key observations are:}
\begin{itemize}[noitemsep, topsep=2pt, leftmargin=*]
\item As $K$ increases, the interquartile range and whisker length shrink markedly, and the empirical optima concentrate near $V_{\mathrm{ref}}$, consistent with diminishing estimation error (cf.\ Theorem~\ref{thm: NN second approx}).
\item The tail probabilities $p_{\varepsilon}$ decrease monotonically with $K$ (e.g.\ $p_{1\%}$ drops from $0.67$ to $0.00$), providing a direct empirical proxy for convergence in probability of the empirical optimum.
\end{itemize}

\begin{figure}[hbt]
\centering
\begin{subfigure}[c]{0.48\textwidth}
\centering
\includegraphics[width=0.8\linewidth]{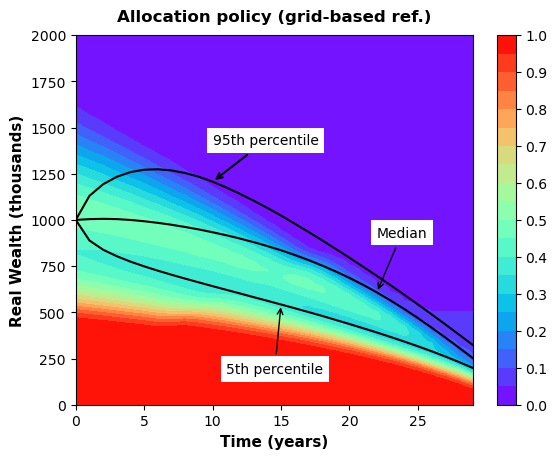}
\caption{Allocation policy (grid-based reference)}
\label{fig:num_pde_P_heatmap}
\end{subfigure}
\hfill
\begin{subfigure}[c]{0.48\textwidth}
\centering
\includegraphics[width=0.8\linewidth]{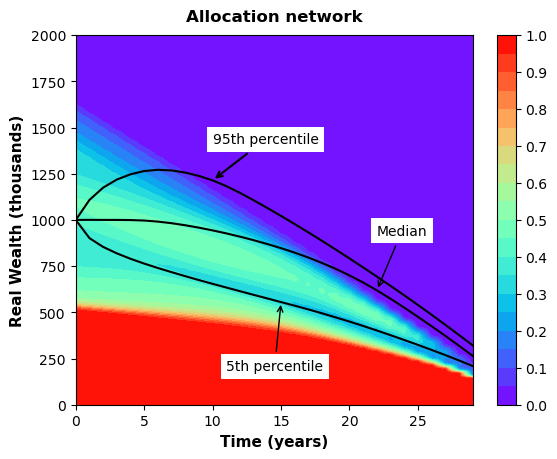}
\caption{Allocation policy (NN)}
\label{fig:num_nn_P_heatmap}
\end{subfigure}
%\vspace{0.4em}
\begin{subfigure}[c]{0.48\textwidth}
\centering
\includegraphics[width=0.8\linewidth]{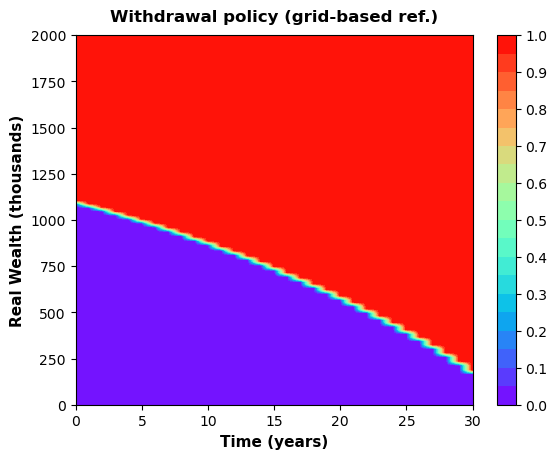}
\caption{Withdrawal policy (grid-based reference)}
\label{fig:num_pde_Q_heatmap}
\end{subfigure}
\hfill
\begin{subfigure}[c]{0.48\textwidth}
\centering
\includegraphics[width=0.8\linewidth]{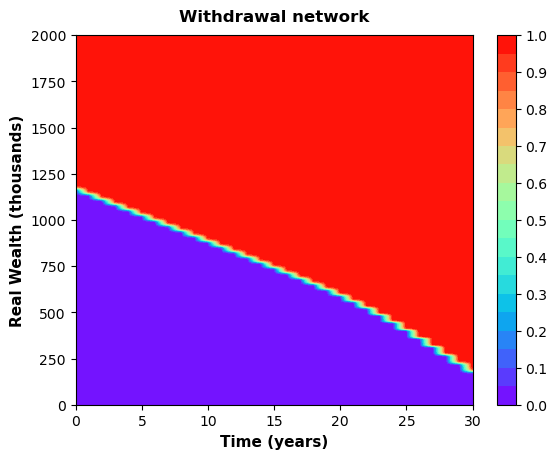}
\caption{Withdrawal policy (NN)}
\label{fig:num_nn_Q_heatmap}
\end{subfigure}
%\vspace*{-0.4cm}
\caption{Policy heatmap comparison.}
\label{fig:num_policy_heatmaps}
%\vspace*{-0.5cm}
\end{figure}

\subsection{Policy structure}
\label{ssc:num_policy_struct}
To complement the value function convergence results in Subsection~\ref{ssc:num_convergence}, we compare the learned feedback maps against the grid-based reference policies.
For a representative trained network (capacity $(L^{(n^\star)},\nu^{(n^\star)})=(2,5)$ and $K=2.56\times 10^5$), we evaluate the NN withdrawal map $q(t_m^-,w)$ and the risky-asset weight (first component) of the NN allocation map $p(t_m^+,w)$ over a uniform $(t,w)$ plotting grid. We then compare these surfaces to the corresponding grid-based reference policies computed in Subsection~\ref{ssc:num_reference}.

\begin{wrapfigure}{r}{0.5\textwidth}
%\vspace*{-0.4cm}
\centering
\includegraphics[width=0.8\linewidth]{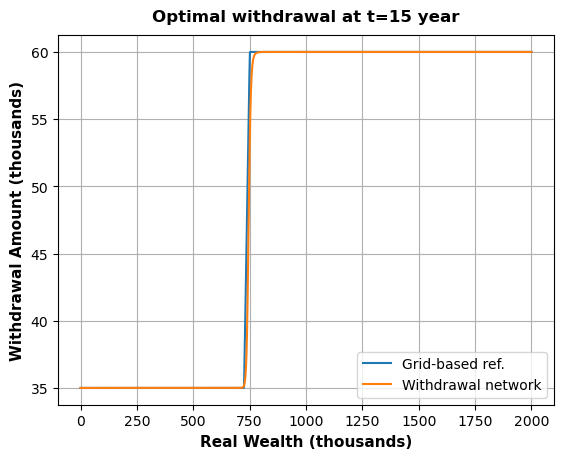}
\caption{Withdrawal slice at $t=15$ years.}
\label{fig:num_withdrawal_slice}
%\vspace*{-0.5cm}
\end{wrapfigure}
Figure~\ref{fig:num_policy_heatmaps} (Panels (a) and (b)) shows that the NN allocation heatmap closely tracks the reference allocation structure across the relevant wealth region, including along typical realized wealth percentiles. {\myblue{More importantly, the withdrawal maps in Figure~\ref{fig:num_policy_heatmaps} (Panels (c) and (d)) exhibit a pronounced quasi--bang--bang pattern: withdrawals concentrate near the bounds $q_{\min}$ and $q_{\max}$, with a narrow transition region along a time-dependent boundary. The NN approximation captures this transition boundary region well, while slightly smoothing the change over a thin band, as expected when a continuous function class approximates a reference map with steep gradients or possible discontinuities.}}

{\myblue{This qualitative structure is compatible with Assumption~\ref{ass:rr_regularity}~(R4): even if the true optimal withdrawal map is not globally continuous in wealth, any nonsmoothness would be concentrated on a lower-dimensional boundary in $(t,w)$. In diffusion-with-jumps return models such as \eqref{eq:dGS}, the one-period return law admits a density, making it plausible (and assumed in (R4)) that the controlled state hits such boundaries with probability zero at intervention times.}}
The slice plot in Figure~\ref{fig:num_withdrawal_slice} provides a 1D view of the same phenomenon at a fixed time. The NN policy closely matches the grid-based threshold structure, with a narrow transition region reflecting {\myblue{approximation of a steep transition in the reference map by a continuous network.}}\footnote{{\myblue{A similar quasi--bang--bang pattern for constrained withdrawal controls has been reported in related numerical studies (see, e.g.\ \cite{forsyth2022stochastic}); here our emphasis is on convergence diagnostics for NN-parametrized policies and on out-of-sample robustness.}}}

\subsection{Out-of-sample robustness}
\label{ssc:num_oos}
The convergence theorems in Section~\ref{sec:conv_analysis} are stated for the empirical optimum
$\widehat V_{NN}(w_0,t_0^-;\mathcal{Y}_K)$ (computed on the same dataset entering the empirical objective).
As a robustness check against potential overfitting concerns in finite samples, we also evaluate each trained policy on a single common independent out-of-sample dataset of size $K_{\mathrm{test}}=2.56\times 10^6$ scenario paths, generated under the same calibrated input law. Concretely, for run $j$ with trained parameters $(\xi^{(j)},\theta^{(j)}_{q,n},\theta^{(j)}_{p,n})$, we compute an out-of-sample value estimate by plugging the trained optimizer into the sample-average functional on the test set:
\[
\widehat V^{(j)}_{n,K_{\mathrm{test}}}
:=
\frac{1}{K_{\mathrm{test}}}\sum_{k=1}^{K_{\mathrm{test}}}
\mathcal{H}\!\left(\xi^{(j)},\,S\!\left(\widehat{\mathcal{P}}^{(j)}_{n},Y_{\mathrm{test}}^{(k)}\right)\right),
\qquad
\widehat{\mathcal{P}}^{(j)}_{n}:=\big(\widehat q_n(\cdot;\theta^{(j)}_{q,n}),\widehat p_n(\cdot;\theta^{(j)}_{p,n})\big),
\]
and summarize errors relative to $V_{\mathrm{ref}}$ in the same way as \eqref{eq:num_p_tail}, i.e.
\begin{equation}
\label{eq:ptest_var}
p^{\mathrm{test}}_{\varepsilon}
:=
\frac{1}{N_{\mathrm{run}}}\sum_{j=1}^{N_{\mathrm{run}}}
\mathbf{1}\Big\{\big|\widehat V^{(j)}_{n,K_{\mathrm{test}}}-V_{\mathrm{ref}}\big|>\varepsilon\,|V_{\mathrm{ref}}|\Big\}.
\end{equation}
Because $K_{\mathrm{test}}$ is large and common across runs, the MC noise in $\widehat V^{(j)}_{n,K_{\mathrm{test}}}$ is small and the cross-run dispersion is dominated by training variability rather \mbox{than test-set noise.}

Table~\ref{tab:num_capacity_conv_oos} reports out-of-sample results as NN capacity increases (with $K=2.56\times 10^5$ fixed), and Table~\ref{tab:num_sample_conv_oos} reports out-of-sample results as $K$ increases (with capacity fixed at $(L^{(n^\star)},\nu^{(n^\star)})=(2,5)$). The qualitative conclusions mirror the in-sample convergence results in Subsection~\ref{ssc:num_convergence}: larger capacity reduces approximation error, and larger $K$ reduces estimation/training variability, with the out-of-sample tail probabilities rapidly collapsing toward zero at moderate sample sizes.

\begin{table}[!hbt]
%\vspace*{-0.35cm}
\caption{Out-of-sample evaluation as NN capacity increases (fixed $K=2.56\times 10^{5}$, $K_{\mathrm{test}}=2.56\times 10^{6}$, $N_{\mathrm{run}}=100$; $V_{\mathrm{ref}}=1605.22$);  $p^{\mathrm{test}}_{\varepsilon}$
defined in \eqref{eq:ptest_var}.}
\label{tab:num_capacity_conv_oos}
%\vspace*{-0.2cm}
\centering
\setlength{\tabcolsep}{6pt}
%\small
\begin{tabular}{c|c|ccccc}
\hline
$(L^{(n)},\nu^{(n)})$ & Mean (Std) & $p^{\mathrm{test}}_{0.5\%}$ & $p^{\mathrm{test}}_{1\%}$ & $p^{\mathrm{test}}_{1.5\%}$ & $p^{\mathrm{test}}_{2\%}$ & $p^{\mathrm{test}}_{2.5\%}$ \\
\hline
$(1,2)$ & $1587.50\ (5.17)$ & $1.00$ & $0.65$ & $0.06$ & $0.00$ & $0.00$ \\
\hline
$(1,5)$ & $1591.86\ (5.81)$ & $0.82$ & $0.32$ & $0.00$ & $0.00$ & $0.00$ \\
\hline
$(2,5)$ & $1600.76\ (4.58)$ & $0.28$ & $0.02$ & $0.00$ & $0.00$ & $0.00$ \\
\hline
\end{tabular}
%\vspace*{-0.25cm}
\end{table}

\begin{table}[!hbt]
%\vspace*{-0.35cm}
\caption{Out-of-sample evaluation as training sample size $K$ increases (fixed capacity $(L^{(n^\star)},\nu^{(n^\star)})=(2,5)$, $K_{\mathrm{test}}=2.56\times 10^{6}$, $N_{\mathrm{run}}=100$; $V_{\mathrm{ref}}=1605.22$); $p^{\mathrm{test}}_{\varepsilon}$ is defined in \eqref{eq:ptest_var}.}
\label{tab:num_sample_conv_oos}
%\vspace*{-0.2cm}
\centering
\setlength{\tabcolsep}{6pt}
%\small
\begin{tabular}{c|c|ccccc}
\hline
$K$ & Mean (Std) & $p^{\mathrm{test}}_{0.5\%}$ & $p^{\mathrm{test}}_{1\%}$ & $p^{\mathrm{test}}_{1.5\%}$ & $p^{\mathrm{test}}_{2\%}$ & $p^{\mathrm{test}}_{2.5\%}$ \\
\hline
$2.56\times 10^{3}$ & $1589.35\ (7.02)$ & $0.91$ & $0.37$ & $0.11$ & $0.03$ & $0.02$ \\
\hline
$2.56\times 10^{4}$ & $1604.41\ (0.78)$ & $0.00$ & $0.00$ & $0.00$ & $0.00$ & $0.00$ \\
\hline
$2.56\times 10^{5}$ & $1606.06\ (0.23)$ & $0.00$ & $0.00$ & $0.00$ & $0.00$ & $0.00$ \\
\hline
\end{tabular}
%\vspace*{-0.25cm}
\end{table} 
\section{Conclusion and future work}
\label{sec:conclusion}
We developed a neural-network approximation framework for discrete-intervention risk--reward stochastic control with constrained two-step feedback policies, where the optimal feedback maps may be discontinuous in the state.
By combining constraint-enforcing output maps with a moving-input stability argument (based on $\mathbb{P}$-null discontinuity sets), we established an end-to-end convergence pipeline: NN approximation of admissible controls propagates through the controlled recursion and a broad class of scalarized risk--reward objectives, and the resulting sample-average training problem converges.
In particular, the empirical optimum converges in probability to the true optimal value as NN capacity and the training sample size increases.
Numerical experiments on a representative decumulation problem corroborate the predicted convergence-in-probability behavior, show excellent agreement between NN and grid-based reference policies (including threshold structure), and indicate robust out-of-sample performance.

Future work includes relaxing bounded-state/compact-domain assumptions, extending the analysis beyond pre-commitment to time-consistent dynamic risk criteria, and studying richer state features and higher-dimensional action spaces.

\appendix
\section{Proof of Lemma~\ref{lem: NN composition}}
\label{app: NN composition}
By Theorem~\ref{thm: convergence of F}, we have $F_{n}(X)\xrightarrow{P} f(X)$ as $n\to\infty$.
By assumption, $\mathbb{P}\!\left(f(X)\in D_\psi\right)=0$, i.e.\ $\psi$ is continuous at $f(X)$ almost surely.
Hence, by the (extended) continuous mapping theorem \cite{Billingsley1995,kallenberg2002},
$\psi(F_{n}(X))\xrightarrow{P}\psi(f(X))$.
\section{Proof of Lemma~\ref{lem: activation function}}
\label{app: activation function}
Fix $\delta\in(0,(b-a)/2)$ and define
$\Pi_\delta:[a,b]^N\to(a,b)^N$ by
\[
(\Pi_\delta(y))_j := \min\{b-\delta,\max\{a+\delta,y_j\}\}, \qquad j=1,\ldots,N.
\]
Then $\Pi_\delta\circ g$ takes values in $[a+\delta,b-\delta]^N\subset(a,b)^N$ and
$\Pi_\delta(g(X))\to g(X)$ a.s.\ as $\delta\downarrow0$. Using the measurable right inverse, set
$h_\delta := \psi^{-1}_r\circ\Pi_\delta\circ g:\mathbb{R}^{\nu_0}\to\mathbb{R}^N$,
which is Borel measurable. By Theorem~\ref{thm: convergence of F}, for each fixed $\delta$ there exists a sequence
$\{F_{n,\delta}\}_{n\in\mathbb{N}}$ with $F_{n,\delta}\in \mathcal{Q}_{n}$ such that
$F_{n,\delta}(X)\xrightarrow{P} h_\delta(X)$ as $n\to\infty$.
By continuity of $\psi$ and the continuous mapping theorem,
\[
\psi\!\left(F_{n,\delta}(X)\right)\xrightarrow{P}\psi\!\left(h_\delta(X)\right)
=\Pi_\delta(g(X)),
\qquad\text{as }n\to\infty.
\]
Now choose a deterministic sequence $\{\delta_j\}_{j\in\mathbb{N}}$ with $\delta_j\downarrow0$.
For each $j$, by the above convergence with $\delta=\delta_j$, there exists an index $n(j)$ such that for all
$n\ge n(j)$,
\[
\mathbb{P}\!\left(\big\|\psi(F_{n,\delta_j}(X))-\Pi_{\delta_j}(g(X))\big\|>\tfrac{1}{j}\right)<\tfrac{1}{j}.
\]
Choose $n(j)$ increasing in $j$ and define $j(n):=\max\{j:\ n(j)\le n\}$, so $j(n)\to\infty$ as
$n\to\infty$. Set
\[
F_{n}:=F_{n,\delta_{j(n)}}\in\mathcal{Q}_{n}.
\]
Then $\psi(F_{n}(X))-\Pi_{\delta_{j(n)}}(g(X))\xrightarrow{P}0$ as $n\to\infty$.
Finally, by the triangle inequality,
\[
\big\|\psi(F_{n}(X))-g(X)\big\|
\le
\big\|\psi(F_{n}(X))-\Pi_{\delta_{j(n)}}(g(X))\big\|
+
\big\|\Pi_{\delta_{j(n)}}(g(X))-g(X)\big\|.
\]
The first term converges to $0$ in probability by construction;
the second converges to $0$ a.s.\ (hence in probability) since $\delta_{j(n)}\downarrow0$.
This completes the proof.

\section{Proof of Theorem~\ref{thm: convergence of q}}
\label{app: convergence of q}
Recall $\mathrm{range}(w)=\max(\min(q_{\max},w)-q_{\min},0)$ from \eqref{eq: psi_q}.
Since $q^\ast$ is admissible, $q^\ast(t,w)\in[q_{\min},\,q_{\min}+\mathrm{range}(w)]$ for all $(t,w)$.
Define the normalized target
\[
u^\ast(t,w)
:=
\begin{cases}
\dfrac{q^\ast(t,w)-q_{\min}}{\mathrm{range}(w)}, & \mathrm{range}(w)>0,\\[6pt]
0, & \mathrm{range}(w)=0.
\end{cases}
\]
Then $u^\ast$ is measurable and $u^\ast(t,w)\in[0,1]$.

Apply Lemma~\ref{lem: activation function} with $N=1$, $[a,b]=[0,1]$, $g:=u^\ast$, and $\psi:=\sigma$.
This yields a sequence of scalar FNN outputs $\{z_{n}(\cdot)\}_{n\in\mathbb{N}}$ such that
$\sigma(z_{n}(X))\xrightarrow{P}u^\ast(X)$ as $n\to\infty$.
Define
\[
\widehat q(t',w';\theta_{q,n}^\ast)
:=
q_{\min}+\mathrm{range}(w')\,\sigma(z_{n}(t',w'))
\equiv
\psi_q\!\big(w', z_{n}(t',w')\big).
\]
Since $(w,u)\mapsto q_{\min}+\mathrm{range}(w)\,u$ is continuous and $w$ is part of the input,
the continuous mapping theorem implies
\[
\widehat q(X;\theta_{q,n}^\ast)
\xrightarrow{P}
q_{\min}+\mathrm{range}(w(X))\,u^\ast(X)
=
q^\ast(X),
\qquad \text{as } n\to\infty.
\]

\section{Proof of Theorem~\ref{thm: convergence of p}}
\label{app: convergence of p}
We first note that $p^\ast(X)\in\mathcal{Z}_p$.
For $\delta>0$, define the interiorized simplex map $\Pi_\delta:\mathcal{Z}_p\to\mathcal{Z}_p$ by
\[
\Pi_\delta(y):=\frac{y+\delta \mathbf{1}}{1+d_a\delta},
\]
so that $\Pi_\delta(y)\in\mathcal{Z}_p$ and each component of $\Pi_\delta(y)$ is strictly positive.
In particular, $\Pi_\delta(p^\ast(X))$ lies in the open simplex and
$\Pi_\delta(p^\ast(X))\to p^\ast(X)$ almost surely as $\delta\downarrow 0$.

On the open simplex, the softmax map $\psi_p$ in \eqref{eq: psi_p} admits the measurable right inverse
\[
\psi_{p,r}^{-1}(y):=(\log y_i)_{i=1}^{d_a},
\qquad y\in\mathcal{Z}_p,\ y_i>0,
\]
since $\sum_{i=1}^{d_a} y_i=1$ implies $\psi_p(\log y)=y$.
Fix $\delta>0$ and define the measurable target
\[
h_\delta(\varphi):=\psi_{p,r}^{-1}\!\big(\Pi_\delta(p^\ast(\varphi))\big)\in\mathbb{R}^{d_a},
\qquad \varphi\in\mathcal{D}_\phi.
\]
By Theorem~\ref{thm: convergence of F} there exists a sequence of FNNs
$\widetilde p_{n,\delta}(\cdot)\in \mathcal{Q}_{n}$ such that
$\widetilde p_{n,\delta}(X)\xrightarrow{P} h_\delta(X)$ as $n\to\infty$.
Since $\psi_p$ is continuous, the continuous mapping theorem yields
\[
\psi_p\!\left(\widetilde p_{n,\delta}(X)\right)
\xrightarrow{P}
\psi_p\!\left(h_\delta(X)\right)
=
\Pi_\delta(p^\ast(X)),
\qquad \text{as } n\to\infty.
\]

To pass to the boundary, fix a deterministic sequence $\{\delta_j\}_{j\in\mathbb{N}}$ with $\delta_j\downarrow 0$.
For each $j$, since $\psi_p(\widetilde p_{n,\delta_j}(X))\xrightarrow{P}\Pi_{\delta_j}(p^\ast(X))$ as $n\to\infty$,
we can choose an index $n(j)$ such that
\[
\mathbb{P}\!\left(
\big\|\psi_p(\widetilde p_{n(j),\delta_j}(X))-\Pi_{\delta_j}(p^\ast(X))\big\|>\tfrac{1}{j}
\right)<\tfrac{1}{j}.
\]
Define a full sequence $\{\widetilde p_{n}\}_{n\in\mathbb{N}}$ by setting, for each $j$,
\[
\widetilde p_{n}:=\widetilde p_{n(j),\delta_j}
\quad\text{for all } n(j)\le n < n(j+1),
\]
and define $\widehat p(\varphi;\theta_{p,n}^\ast):=\psi_p(\widetilde p_{n}(\varphi))$.
By nesting, $\widetilde p_{n}\in\mathcal{Q}_{n}$ for all $n$. Moreover,
\[
\widehat p(X;\theta_{p,n}^\ast)
=
\psi_p(\widetilde p_{n}(X))
\xrightarrow{P}
\Pi_{\delta_{j(n)}}(p^\ast(X)),
\qquad \text{as } n\to\infty,
\]
with $j(n)\to\infty$.
Finally, since $\Pi_{\delta_{j(n)}}(p^\ast(X))\to p^\ast(X)$ almost surely (hence in probability),
the triangle inequality yields $\widehat p(X;\theta_{p,n}^\ast)\xrightarrow{P}p^\ast(X)$.

\section{Proof of Lemma~\ref{lem: convergence of objective}}
\label{app: convergence of objective}
Write $\mathbb{E}[\cdot]$ for $\mathbb{E}^{w_0,t_0^-}_{\mathcal{P}}[\cdot]$ when the control is clear from context.
Define the NN-induced and optimal performance vectors
\[
S^{(n)}:=S(\Theta_{n}^\ast,Y),
\qquad
S^\ast:=S(\mathcal{P}^\ast,Y).
\]
Since $S(\cdot,Y)$ is a finite-dimensional vector built from the intervention-time sequence
\\
$\{W(t_m^\pm),q(t_m^-,W(t_m^-)),p(t_m^+,W(t_m^+))\}$, 
Lemmas~\ref{lem: convergence of q amount}--\ref{lem: convergence of W} imply componentwise convergence in probability, hence
\[
S^{(n)}\xrightarrow{P} S^\ast.
\]
For moment dependence, define
\[
\bar S^{(n)}:=\mathbb{E}\!\left[\Psi(S^{(n)})\right],
\qquad
\bar S^\ast:=\mathbb{E}\!\left[\Psi(S^\ast)\right],
\]
where $\Psi$ is as in \eqref{eq:rr_moment_vector}. By Assumption~\ref{ass:rr_regularity}~(R5), $\Psi$ is continuous on the compact set $\mathcal{Z}$ containing $S(\mathcal{P},Y)$ a.s.\ for all admissible $\mathcal{P}$, hence $\Psi$ is bounded on $\mathcal{Z}$.
By the continuous mapping theorem, $\Psi(S^{(n)})\xrightarrow{P}\Psi(S^\ast)$, and boundedness implies uniform integrability, hence
\[
\bar S^{(n)}=\mathbb{E}[\Psi(S^{(n)})]\longrightarrow \mathbb{E}[\Psi(S^\ast)]=\bar S^\ast.
\]
(If $d_\Psi=0$, the $\bar S$ terms are absent and this step can be ignored.)

Now define the scalarized criterion random variables
\[
\mathcal{H}^{(n)}:=\mathcal{H}\left(\xi, S^{(n)}, \bar S^{(n)}\right),
\qquad
\mathcal{H}^\ast:=\mathcal{H}\left(\xi, S^\ast, \bar S^\ast\right).
\]
Since $(S^{(n)},\bar S^{(n)})\xrightarrow{P}(S^\ast,\bar S^\ast)$ and $\mathcal{H}$ is continuous on
$\Xi\times\mathcal{Z}\times\bar{\mathcal{Z}}$ by Assumption~\ref{ass:rr_regularity}~(R7), the continuous mapping theorem yields
$\mathcal{H}^{(n)}\xrightarrow{P}\mathcal{H}^\ast$.
Moreover, by Assumption~\ref{ass:rr_regularity}~(R7), the integrand $\mathcal{H}$ is bounded on $\Xi\times\mathcal{Z}\times\bar{\mathcal{Z}}$, so $\{\mathcal{H}^{(n)}\}_{n}$ is uniformly integrable. Therefore,
$ \mathbb{E}[\mathcal{H}^{(n)}]\longrightarrow \mathbb{E}[\mathcal{H}^\ast]$,
which is
$V_{NN}(w_0,t_0^-;\xi,\Theta_{n}^\ast)\to V(w_0,t_0^-;\xi,\mathcal{P}^\ast)$.
This concludes the proof.

\section{Proof of Lemma~\ref{lem: NN ULLN}}
\label{app: NN ULLN}
Write $\mathbb{E}[\cdot]$ for expectation under $\mathbb{P}$ conditional on $W(t_0^-)=w_0$. The dependence on the control parameter $\Theta_{n}$ enters only through the measurable map $Y\mapsto S(\Theta_{n},Y)$.
Fix $n$.
Recall from \eqref{eq:barZ}--\eqref{eq:widehat_barZ}
%the moment vector
\[
\bar S(\Theta_{n})
:=
\mathbb{E}\!\big[\Psi(S(\Theta_{n},Y))\big]\in\mathbb{R}^{d_\Psi},
\]
and its sample-average estimator based on i.i.d.\ copies $Y^{(1)},\ldots,Y^{(K)}$,
\[
\widehat{\bar S}_K(\Theta_{n})
:=
\frac{1}{K}\sum_{k=1}^K \Psi\!\left(S^{(k)}(\Theta_{n})\right)\in\mathbb{R}^{d_\Psi}.
\]
%We introduce the intermediate (oracle) empirical objective which uses
We introduce the auxiliary empirical objective which uses
$\bar S(\Theta_{n})$ instead of $\widehat{\bar S}_K(\Theta_{n})$:
\[
\widetilde{V}_{NN}(w_0,t_0^-;\xi,\Theta_{n},\mathcal{Y}_K)
:=
\frac{1}{K}\sum_{k=1}^{K}
\mathcal{H}\!\left(\xi,\ S^{(k)}(\Theta_{n}),\ \bar S(\Theta_{n})\right).
\]
Then, by the triangle inequality,
\begin{equation}
\label{eq:decom}
\sup_{\Theta_{n},\xi\in\Xi}\big|\widehat{V}_{NN}-V_{NN}\big|
\le
\sup_{\Theta_{n},\xi\in\Xi}\big|\widehat{V}_{NN}-\widetilde{V}_{NN}\big|
+
\sup_{\Theta_{n},\xi\in\Xi}\big|\widetilde{V}_{NN}-V_{NN}\big|.
\end{equation}
%\medskip
%\noindent\emph{Step 1: Uniform LLN for the oracle objective.}
\noindent\emph{Step 1: The second term on the rhs of \eqref{eq:decom}}
For fixed $n$, the function class
\[
\Big\{y\mapsto \mathcal{H}\big(\xi,S(\Theta_{n},y),\bar S(\Theta_{n})\big)
:\ (\xi,\Theta_{n})\in\Xi\times\mathbb{R}^{\vartheta^{(n)}}\Big\}
\]
is uniformly bounded by Assumption~\ref{ass:rr_regularity}~(R5)/(R7), since $S(\Theta_{n},Y)\in\mathcal{Z}$ a.s.\ for all admissible controls and $\bar S(\Theta_{n})\in\bar{\mathcal{Z}} =\mathrm{conv}(\Psi(\mathcal{Z}))$.
Therefore a suitable ULLN (e.g.\ \cite[Thm.~4.3]{tsang2020deep} for fixed architecture) yields
\[
\sup_{\Theta_{n},\xi\in\Xi}\big|\widetilde{V}_{NN}-V_{NN}\big|\xrightarrow{P}0,
\qquad \text{as }K\to\infty.
\]
%\medskip
%\noindent\emph{Step 2: Uniform control of the plug-in effect.}
\noindent\emph{Step 2: The first term on the rhs of \eqref{eq:decom}}
By (R5)/(R7), $\Psi(\mathcal{Z})$ is compact and $\bar{\mathcal{Z}}=\mathrm{conv}(\Psi(\mathcal{Z}))$ is compact.
Moreover, $\widehat{\bar S}_K(\Theta_{n})\in\bar{\mathcal{Z}} =\mathrm{conv}(\Psi(\mathcal{Z}))$ for every $(K,\Theta_{n})$.
Since $\mathcal{H}$ is continuous on $\Xi\times\mathcal{Z}\times\bar{\mathcal{Z}}$ and the latter is compact whenever $\Xi$ is compact
(or more generally if $\mathcal{H}$ is uniformly continuous on $\Xi\times\mathcal{Z}\times\bar{\mathcal{Z}}$), $\mathcal{H}$ is uniformly continuous in its third argument on this domain. Hence, for any $\varepsilon>0$ there exists \mbox{$\delta>0$ such that}
\[
\| \bar s-\bar s'\|\le\delta
\quad\Longrightarrow\quad
\sup_{(\xi,s)\in\Xi\times\mathcal{Z}}
\big|\mathcal{H}(\xi,s,\bar s)-\mathcal{H}(\xi,s,\bar s')\big|
\le \varepsilon.
\]
Therefore, on the event $\{\sup_{\Theta_{n}}\|\widehat{\bar S}_K(\Theta_{n})-\bar S(\Theta_{n})\|\le \delta\}$,
\[
\sup_{\Theta_{n},\xi\in\Xi}\big|\widehat{V}_{NN}-\widetilde{V}_{NN}\big|
\le \varepsilon.
\]
Finally, since $\Psi(S(\Theta_{n},Y))$ is uniformly bounded and the architecture is fixed, a ULLN applied to the class
$\{y\mapsto \Psi(S(\Theta_{n},y))\}$ implies
\[
\sup_{\Theta_{n}}\big\|\widehat{\bar S}_K(\Theta_{n})-\bar S(\Theta_{n})\big\|\xrightarrow{P}0.
\]
Combining the preceding displays yields
$\sup_{\Theta_{n},\xi\in\Xi}\big|\widehat{V}_{NN}-\widetilde{V}_{NN}\big|\xrightarrow{P}0$.
Together with Step 1, this proves the claimed ULLN for $\widehat V_{NN}$.

\section{Convergence table for grid-based method in Subsection~\ref{ssc:num_reference}}
\label{app:num_reference}
\begin{table}[hbt]
\vspace*{-0.4cm}
\caption{Convergence test for EW--CVaR problem with $\alpha=0.05$ and $\gamma=1.0$.
Grid size denotes the discretization used in the numerical scheme $\left( N_{y} \times N_{b} \right)$, where $N_{y}$ is the number of nodes in the risky log-domain and $N_{b}$ is the number of nodes in the risk-free domain; 
at each refinement level, the intervention-time searches use a wealth grid with $N_w=4N_y$ nodes, and the outer $\xi$-search uses $N_\xi=N_b$ nodes.
Units: thousands of dollars (real).}
\label{tab: PDE results table}
% \vspace{0.2cm}
\centering
\setlength{\tabcolsep}{6pt}
\begin{tabular}{c|c|c|c|c}
%\begin{tabular}{|>{\raggedright}p{7cm}|>{\centering}p{1.6cm}|}
\hline
Grid size & Value function & $E\left[ \sum_{m}q_{m}\right]/\left(M+1\right)$ & $\text{CVaR}_{5\%}\left[W_{T}\right]$ & $\xi^{\ast}$ \\
% \hline
% $\left(256 \times 167\right)$ & $1598.59$ & $50.52$ & $32.61$ & $123.36$ \\
% \hline
% $\left(512 \times 333\right)$ & $1601.33$ & $50.94$ & $22.33$ & $124.32$ \\
% \hline
% $\left(1024 \times 665\right)$ & $1604.04$ & $50.76$ & $30.47$ & $130.50$ \\
% \hline
% PDE $\left(1024 \times 665\right)$ & $1600.64$ & $52.28$ & $-19.96$ & $40.80$ \\
% \hline

\hline
$\left(512 \times 512\right)$ & $1600.08$ & $50.96$ & $20.41$ & $128.10$ \\
\hline
$\left(1024 \times 1024\right)$ & $1604.19$ & $50.80$ & $29.52$ & $129.40$ \\
\hline
$\left(2048 \times 2048\right)$ & $1605.22$ & $50.76$ & $31.60$ & $129.90$ \\
\hline
\end{tabular}
% \vspace{+0.2cm}
\end{table}

%\bibliographystyle{plain}
%\bibliography{paperbib_cc}
%\setlength{\bibsep}{0pt plus 0.3ex}

\vfill
\end{document}